\documentclass[reprint, footinbib,  aps, pra, floatfix,twocolumn,superscriptaddress]{revtex4-2}
\usepackage[hidelinks]{hyperref}

\usepackage[english]{babel}
\usepackage[utf8]{inputenc}
\usepackage{amsthm}
\usepackage{amsfonts,fixmath}
\usepackage{graphicx}
\usepackage{dcolumn}
\usepackage{bm}
\usepackage{physics}

\usepackage{xcolor}
\usepackage{mathtools}

\newtheorem{theorem}{Theorem}[section]
\newtheorem{corollary}{Corollary}[theorem]
\newtheorem{lemma}[theorem]{Lemma}
\newtheorem{proposition}[theorem]{Proposition}
\newtheorem{definition}{Definition}[theorem]

\newcommand{\ketbraind}[3]{\ket{#1}_#2\! \bra{#3}}

\newcommand{\ceil}[1]{\left\lceil #1 \right\rceil}
\newcommand{\floor}[1]{\left\lfloor #1 \right\rfloor}

\begin{document}

\begin{abstract}
Variational quantum algorithms (VQAs) have enabled a wide range of applications on near-term quantum devices. However, their scalability is fundamentally limited by barren plateaus, where the probability of encountering large gradients vanishes exponentially with system size. In addition, noise induces barren plateaus, deterministically flattening the cost landscape.
Dissipative quantum algorithms that leverage nonunitary dynamics to prepare quantum states via engineered cooling offer a complementary framework with remarkable robustness to noise. We demonstrate that dissipative quantum algorithms based on non-unital channels can avoid both unitary and noise-induced barren plateaus. Periodically resetting ancillary qubits actively extracts entropy from the system, maintaining gradient magnitudes and enabling scalable optimization. We provide analytic conditions ensuring they remain trainable even in the presence of noise. Numerical simulations confirm our predictions and illustrate scenarios where unitary algorithms fail but dissipative algorithms succeed. Our framework positions dissipative quantum algorithms as a scalable, noise-resilient alternative to traditional VQAs.
\end{abstract}
\title{Scaling Quantum Algorithms via Dissipation: Avoiding Barren Plateaus}

\author{Elias Zapusek}
\email{zapuseke@ethz.ch}
\affiliation{Institute for Quantum Electronics, ETH Z\"{u}rich, 8093 Z\"{u}rich, Switzerland}
\author{Ivan Rojkov}
\affiliation{Institute for Quantum Electronics, ETH Z\"{u}rich, 8093 Z\"{u}rich, Switzerland}
\author{Florentin Reiter}
\affiliation{Institute for Quantum Electronics, ETH Z\"{u}rich, 8093 Z\"{u}rich, Switzerland}
\affiliation{Fraunhofer  Institute for Applied Solid State Physics IAF, Tullastr. 72, 
79108 Freiburg, Germany}

\date{\today}
\maketitle

\section{Introduction}
Variational quantum algorithms (VQAs) have significantly expanded the applications of quantum computers \cite{peruzzo_variational_2014,cerezo_variational_2021}. They encode a computational problem in a cost function, and by optimizing it variationally countless new applications can be explored. Their relatively short depth makes them suitable for noisy present-day devices \cite{peruzzo_variational_2014}. However, they suffer from fundamental scalability barriers known as barren plateaus, where the gradient exponentially concentrates at zero; meaning that the probability of encountering a nontrivial gradient decreases exponentially with the size of the system \cite{mcclean_barren_2018,sharma_trainability_2022,arrasmith_effect_2021,cerezo_cost_2021,holmes_connecting_2022,cerveromartin_barren_2023,cerezo_does_2024,crognaletti_estimates_2024,larocca_barren_2025}. Barren plateaus relate to the size of the space the ansatz explores and can be caused by circuit expressiveness \cite{mcclean_barren_2018,holmes_connecting_2022,zhang_absence_2024}, global cost functions \cite{cerezo_cost_2021} and entanglement \cite{wiersema_measurementinduced_2023,sack_avoiding_2022}. The presence of noise exacerbates the issue resulting in deterministic concentration, referred to as noise-induced barren plateaus \cite{wang_noise-induced_2021,schumann_emergence_2024,singkanipa_beyond_2025}.

Dissipative quantum algorithms offer a complementary approach to purely unitary dynamics \cite{poyatos_quantum_1996,verstraete_quantum_2009,reiter_scalable_2016, polla_Quantum_2021,cubitt_Dissipative_2023,chen_efficient_2023,ding_single-ancilla_2024,mi_stable_2024,ulcakar_generalized_2024,cobos_noiseaware_2024,ilin_dissipative_2025,zapusek_variational_2025,lin_dissipative_2025}. They offer advantages such as driving to the computational output as a steady state \cite{verstraete_quantum_2009}, expanding the class of preparable states \cite{chen_efficient_2023,reiter_Engineering_2021}, and enabling alternative algorithmic strategies that can tackle problems beyond the reach of unitary approaches \cite{chen_local_2025}. In the context of variational algorithms, dissipation has been shown to alleviate barren plateaus that arise due to the globality of the cost functions \cite{sannia_engineered_2024}. Dissipative quantum algorithms, which exploit interactions with an environment to drive a quantum system toward a target state through cooling dynamics, are highly resilient to noise \cite{kastoryano_dissipative_2011,morigi_dissipative_2015, polla_Quantum_2021,mi_stable_2024}. This raises a natural question: can the robustness of dissipative ground state preparation be extended to parameterized quantum algorithms?

In this work, we demonstrate that dissipative quantum algorithms leveraging non-unital quantum channels are free from both noise-induced and unitary barren plateaus. By periodically resetting ancillary qubits during the algorithm's operation, entropy is actively removed from the system, preventing gradient concentration and enabling efficient optimization. Our approach builds upon the techniques developed in Ref.~\cite{mele_noise-induced_2024}, which established that non-unital noise can mitigate barren plateaus. Our work is part of a broader line of work leveraging dissipation engineering to enhance quantum algorithms \cite{verstraete_quantum_2009,zapusek_variational_2025,halimeh_stabilizing_2022} and, unlike  Ref.~\cite{mele_noise-induced_2024} considers reset as a resource, demonstrating the resilience of dissipative algorithms. 
The variance of the cost function in dynamical circuits can be lower bounded, thereby indicating the existence of trainable parameters in principle \cite{deshpande_dynamic_2024}. In contrast, we provide explicit architectural criteria and pinpoint which specific gradients remain robust against concentration. This constructive approach offers practical guidance for designing circuits with guaranteed trainability.
Importantly, not all dissipative constructions are robust to noise; we establish when dissipative circuits can scale in the presence of noise and show that this robustness allows dissipative algorithms to solve problems that are inaccessible to noisy unitary circuits. While the resources involved are reminiscent of those required for full quantum error correction, the associated overhead in physical qubits is significantly reduced. Our results establish dissipative quantum algorithms as a scalable alternative to unitary approaches, opening new pathways for robust quantum computation on noisy devices.

\begin{figure*}
    \centering
    \includegraphics[width=\textwidth]{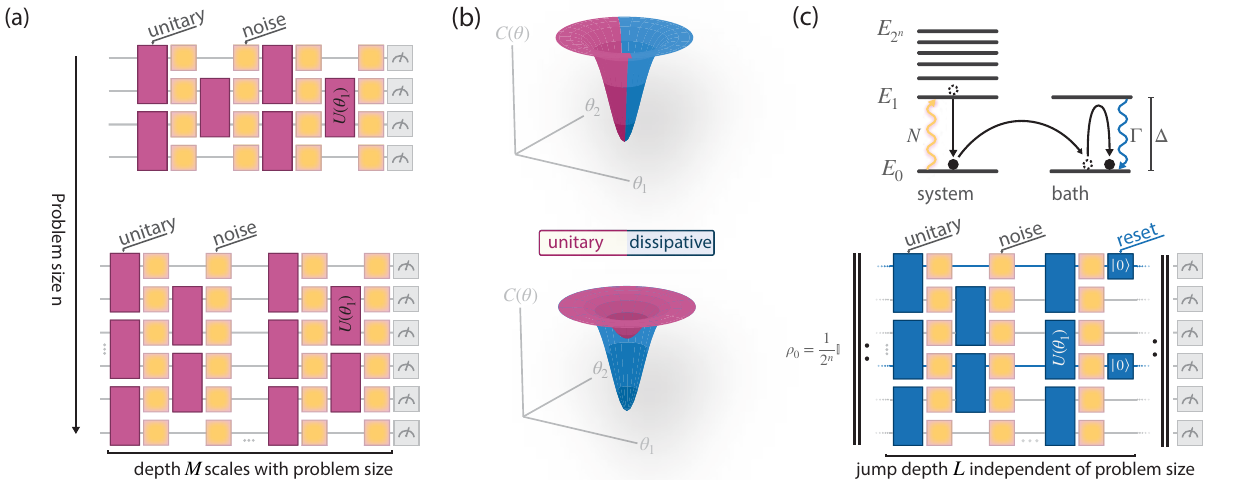}
    \caption{\textbf{Comparison of unitary and dissipative algorithms.} \textbf{(a)} Unitary circuit. For most problems the depth of the unitary circuit must scale with the size of the problem. The increase in depth exponentially suppresses the gradients of the cost function making the circuit untrainable. \textbf{(b)} Cost landscape. Dissipative algorithms can avoid the suppression of gradients. \textbf{(c)} Dissipative circuits. The algorithm works by
weakly coupling a system to a cooled bath. When the energy splitting of the system $E_1-E_0$, matches that of the bath $\Delta$ an excitation of the system can be mapped to the bath and removed by the cooling $\Gamma$. This is noise robust, as if an error re-excites the system it can be removed by the cooling.
In the circuit picture, system and bath evolutions are modeled by unitaries, and the cooling is modeled by periodic reset of qubits.
}
    \label{fig:Maindoublewide}
\end{figure*}
\section{Results:}
\subsection{Framework}\label{sec:Framework}

We investigate random parameterized quantum circuits under the influence of noise, comparing two different circuit ans\"atze as shown in Fig.~\ref{fig:Maindoublewide}. The qubits are arranged on a lattice of dimension $d$, where each qubit is connected to 2$d$ neighbors via unitaries. Panel (a) depicts the unitary circuit, which consists solely of unitary gates and unital noise channels with circuit depth $M$. In contrast, the dissipative circuit shown in panel (c) includes unitary gates, unital noise, and additional non-unital channels---specifically, reset operations. We define $L$ as the depth (that is, the number of layers) between consecutive resets and $M$ as the total number of reset applications. The overall system evolution is described by a quantum channel 
\begin{align} \label{eq:circuit_M}
\Phi^M = \Phi_{\left[ L \cdot M,1\right]},
\end{align}
where
\begin{align}\label{eq:circuit_ab}
\Phi_{\left[ a,b\right]} = \prod_{k=a}^{b} \left(\mathcal{V}^k\circ \mathcal{N}\circ \mathcal{U}^k \circ \mathcal{A}^{\chi(k)}\right).
\end{align}
Here, $\mathcal{V}^j(\cdot) = V^j (\cdot) V^{j\dag}$ represents a layer of single-qubit gates, $\mathcal{N}$, the unital noise, and $\mathcal{A}$ denotes the reset of a subset of qubits. Each layer $\mathcal{U}^j = U^j (\cdot) U^{j\dag}$ consists of a single application of a $d$-dimensional brickwork circuit composed of two-qubit gates. Figure~\ref{fig:circuit_modelmain} provides a schematic representation of the dissipative circuit architecture. We are interested in the average case behavior of these circuits and assume that the unitaries are drawn from a distribution that forms a two-design, i.e., matches the uniform distribution over unitaries up to the second moment \cite{dankert_exact_2009}. Unital noise $\mathcal{N}$ is applied after each layer of two-qubit gates. The noise channel $\mathcal{N}$ is modeled as an $n$-qubit tensor product of contractive, unital single-qubit channels, distinct from the fully depolarizing channel. Under the action of the single-qubit noise channel, the Bloch sphere is mapped to an ellipsoid with axis-specific contraction factors. $D_{\max}$ denotes the largest of these factors, corresponding to the axis least affected by the noise.
In the dissipative circuit, the non-unital channel $\mathcal{A}_{n_r}$ implements periodic resets of $n_r$ equidistant qubits via an amplitude damping process. That is, in a $d$-dimensional lattice of size $n$, the spacing between reset qubits is constrained so that in each direction their separation is at least $\floor{ \left( n/{n_r}\right)^{1/d}}$ and at most $\ceil{ \left( {n}/{n_r}\right)^{1/d}}$. Each reset occurs with probability $q$. The function $\chi(l)$ determines when the resets are applied:
$ \chi(l) = 1 $ if $ l \bmod L = 0$ and $ \chi(l) = 0$ otherwise.
We refer to a block of $L$ layers as a \textit{jump}, marking the interval between consecutive reset operations.

The circuit can depend on a set of variational parameters $\bm{\theta} = (\theta_1, \dots, \theta_m) \in \mathbb{R}^m$, which control a subset of two-qubit gates. Each such gate has the form $\exp(-i\theta_\mu H_\mu)$, where $H_\mu$ are two-local Hermitian operators that satisfy $0 < \|H_\mu\|_\infty \leq 1$. The parameters are optimized by minimizing a cost function of the form
$C(\bm{\theta}) =\Tr(O \rho (\bm{\theta})) =\Tr(O\, \Phi^M\! (\rho_0))$ with observable $O$. The diameter of an observable is the maximum Manhattan distance between any two non-identity Pauli operators within any single Pauli string that appears in the observable's decomposition.
Owing to the left- and right-invariance of Haar random two-qubit gates, the action of the parameterized gates can be absorbed into the surrounding random unitaries.

\begin{figure}
    \centering
    \includegraphics[width = \columnwidth]{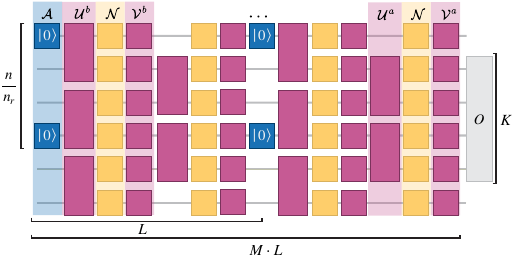}
    \caption{\textbf{Circuit model.} The circuit model $\Phi_{[a,b]}$ in one dimension. The circuit consists of the channel $\mathcal{A}^\chi$ (resetting equidistant qubits), two-qubit unitaries $\mathcal{U}$, noise $\mathcal{N}$ and single-qubit unitaries $\mathcal{V}$.
    Since $\chi(b) = 1$ and $\chi(a) = 0$, a reset is applied in the first displayed layer but not in the last. The jump depth---defined as the number of layers between consecutive resets---is $L = 2$. The fraction of qubits that is reset is described by $n_r/n = 1/3$. The parameter $M$ denotes the total number of reset events (or "jumps") applied throughout the circuit, while $K$ specifies the diameter of the observable $O$.}
    \label{fig:circuit_modelmain}
\end{figure}
Unitary brickwork circuits, illustrated in Fig.~\ref{fig:Maindoublewide}(a), exhibit the barren plateau phenomenon \cite{mcclean_barren_2018,harrow_approximate_2023}, whereby the variance of the cost function gradients is exponentially concentrated with system size:
\begin{align}
    \operatorname{Var}\left[ \partial_\mu C \right] = \mathcal{O}\left(\frac{1}{b^n} \right),
\end{align}
for some constant $b>1$. Since the gradients are also centered around zero, encountering large gradients becomes exponentially unlikely as the system size grows. Resolving these vanishing gradients requires an exponential number of measurement shots, effectively eliminating any potential quantum advantage. 

Noise further amplifies the problem, leading to deterministic concentration \cite{wang_noise-induced_2021}, where the entire cost landscape becomes flat. A phenomenon known as Noise-induced barren plateaus (NIBP). Both gradients and expectation values converge exponentially with circuit depth toward a noise-induced fixed point, rendering optimization infeasible and eliminating the possibility of any smart initialization strategies.

\subsection{Analytic results}\label{sec:Analytic_results}
With the challenge of barren plateaus in mind, we investigate the gradients of dissipative quantum circuits and find that dissipative quantum algorithms exhibit drastically different behavior compared to their unitary counterparts: Periodically resetting a subset of qubits avoids both unitary and noise-induced barren plateaus. We first present our results on the variance of cost functions, as it is instrumental in proving our main result.
\begin{proposition}(Lower bound of the variance) \label{prop:variance_lowerbound_main}
     Consider a circuit as in Sec.~\ref{sec:Framework} defined on a lattice of $n$ qubits and dimension $d$. Every $L$ layers, $n_{r}$ equidistant qubits are reset with probability $q$. Let $O$ be an observable with $\operatorname{diam}(O) = K$. The variance after $M \in \mathbb{N}$ jumps is lower bounded by:
\begin{align}\label{eq:lower_bound_variance_main}
    \operatorname{Var} \left[ \Tr\left(O\, \Phi^M\!(\rho_0)\right)\right] \geq \frac{A}{C^{\Lambda}}\, (D_{\max})^{2 {K^d  \Delta}} \, q^{\,2 \Gamma }.
\end{align}
     $D_\text{max}$ is the largest contraction factor, corresponding to the axis least shrunk by the single-qubit noise channel. $C$ and $A$ are positive constants. Crucially, $\Lambda$, $\Gamma$, and $\Delta$ depend polynomially on $K$ and $L$ but are independent of the depth of the circuit ${L \cdot M}$. If $n_r = \mathcal{O}(n)$, the lower bound is independent of the number of qubits $n$.
\end{proposition}
The proof is given in Sec.~\ref{sec:VarLowerMethods} and further details are given in SI.~\ref{app:Variance_Lower_app}. Periodic qubit resets prevent the concentration of expectation values, as the non-unital reset channels counteract the effects of unital noise. Practically, the non-vanishing lower bound on the variance enables gradient-free optimization strategies \cite{arrasmith_effect_2021} and allows for characterization of the cost function value \cite{peruzzo_variational_2014}. This mechanism underpins our main result on the gradients of dissipative circuits.
\begin{corollary}(Absence of barren plateaus logarithmic depth jumps) \label{cor:AbsenceBarrenLogmain}
    Consider a circuit as described in Section~\ref{sec:Framework}, with jumps of logarithmic depth ${L = \mathcal{O}(\log(n))}$ on a $d$-dimensional lattice. Let $\exp(-i \theta_\mu H_\mu)$ be a parameterized gate located at layer $LM - i$, where ${i = \mathcal{O}(\log(n))}$. If the number of reset qubits scales with system size, that is, $n_r = \mathcal{O}(n)$, and the support of the gate lies within the inverse light cone of a local cost function with bounded support $K = \mathcal{O}(1)$, then the gradient of the cost function with respect to $\theta_\mu$ does not exhibit exponential concentration. We can lower bound the variance of the gradient by 
    \begin{align}
        \operatorname{Var}\left[\partial_\mu C \right] \geq \Omega\left(\frac{1}{\operatorname{poly(n)}}\right).
    \end{align}
\end{corollary}
The statement is proven in Sec.~\ref{sec:AbsenceBarrenSketch} further details are provided in SI.~\ref{app:Barren_Lower_app}.
Periodic resets alter the gradient landscape, preventing concentration. It should be noted that the gradient variance is independent of the depth of the circuit $L \cdot M$.

 In the case where the jumps are of constant depth, $L = \mathcal{O}(1)$, the support of the observable $O$ is bounded, and the gate of which we calculate the gradient is a constant number of layers from the measurement, we can lower bound the variance of the gradient by a constant.

\subsection{Numerical results}
To support our analytical results and explore circuit structures beyond the assumptions of the proof, we perform numerical simulations. These include scenarios where unitaries are correlated across layers, as well as other widely used ans\"atze, such as the QAOA-inspired ansatz, demonstrating that our analytical results hold in highly relevant and practical settings. In SI.~\ref{app:Barren_nonoise} further details on the numerics are given.

We isolate the effect of unitary barren plateaus by considering a setting without noise. In the first scenario, the ansatz we consider has the brickwork structure shown in Fig.~\ref{fig:circuit_modelmain}. It consists of two-qubit bricks of the form:
\begin{align}
    U(\bm{\theta}) =\left[ R_Y(\theta_1) \otimes R_Y(\theta_2) \right] \operatorname{CNOT}\left[R_X(\theta_3) \otimes R_X(\theta_4)\right].
\end{align}
The rotations are defined as $R_A(\theta) = e^{-\frac{i}{2} \theta A }$. 
We consider the architecture in one dimension, with the bricks alternatingly connecting qubits $(1,2)(3,4)\dots(n-1,n)$ and qubits $(2,3)(4,5)\dots(n-2,n-1)$. The ends of the one-dimensional chain are not connected.

In the second scenario, we go beyond the brickwork structure with a QAOA inspired circuit often used in quantum machine learning, state preparation, and optimization tasks \cite{farhi_quantum_2014,verdon_quantum_2019}. The circuit layers consist of gates $R_{ZZ}$ $R_X$ and $R_y$, while $R_{ZZ}$ is applied with even and odd first qubits.

We refer to the \textit{dissipative ansatz} as the version where non-unital amplitude damping channels are added to the two scenarios. This channel is applied to every second qubit every 5 layers, assuming a perfect reset with the output state $\ketbra{0}{0}$.

Using a circuit consisting of $M=40$-layers, we define the cost function $C(\bm{\theta}) = \Tr (O \rho(\bm\theta) )$ for the observable $O = Z_1$, which corresponds to measuring Pauli-$Z$ on the second qubit. The gradient is calculated with respect to the parameter $\theta_\mu$ located in the final layer before measurement.

In Fig.~\ref{fig:Main_unitarybarren2plots} we plot the variance of the gradient. In panel (a), we show the variance of the gradient for all unitaries parameterized individually. Panel (b) goes beyond the assumptions of our proof; here the parameters of the unitaries are repeated every five layers, reminiscent of the circuit shown in Fig.~\ref{fig:Maindoublewide} (c). Both panels qualitatively show the same behavior; the variance of the gradient is exponentially concentrated for the unitary ansatz but not for the dissipative one. Consequently, the unitary ansatz suffers from barren plateaus, while the dissipative ansatz is barren plateau-free. The concentration of the gradient is due to the space the ansatz explores rather than an effect of the noise, since the unitary ansatz is noiseless. 

\begin{figure}
    \centering
    \includegraphics{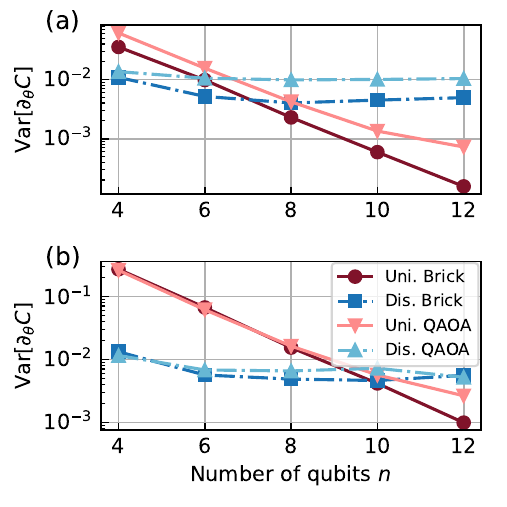}
    \caption{\textbf{Unitary barren plateaus.} We plot the gradient variance for a one-dimensional ansatz with 40 layers. The dissipative ansatz resets every second qubit after 5 layers. In (a), the parameters of all layers are independently distributed, while in (b), they repeat every five layers. The gradient is measured for a local cost function based on the observable $Z$ on the second qubit. While the unitary gradient decays exponentially with system size, the gradient of the dissipative ansatz stays constant.}
    \label{fig:Main_unitarybarren2plots}
\end{figure}
We now numerically show that the dissipative variational quantum circuits can avoid noise-induced barren plateaus. We consider the preparation of toric code ground states, prototypical examples of topologically ordered states, in the presence of noise. We assume qubits arranged on a torus with local connectivity and in the dissipative setting, we additionally assume that each plaquette and vertex are coupled to an ancilla. These ancillae are reset during the operation of the algorithm. The Hamiltonian of the toric code reads \cite{dennis_topological_2002}:
\begin{align}
    H = -\sum_\nu A_{\nu} -\sum_\beta B_{\beta}.
\end{align}
Here, $A_\nu = X_{\nu,1}X_{\nu,2}X_{\nu,3}X_{\nu,4}$ describes the stabilizer associated with plaquette $\nu$ and $B_\beta= Z_{\beta,1}Z_{\beta,2}Z_{\beta,3}Z_{\beta,4}$ the stabilizer associated with vertex $\beta$.

A quantum circuit that, starting from a product state, prepares the ground state of the toric code with quasi-local gates on the lattice must have a depth of at least $\Omega( \sqrt{n})$ where $n$ is the number of qubits in the system due to the Lieb-Robinson bound \cite{dennis_topological_2002,bravyi_lieb-robinson_2006}. We assume that we do not know the form of the unitary circuit that prepares the ground state. Instead, we want to train a variational quantum eigensolver (VQE) with a QAOA ansatz to prepare the ground state. We interleave layers of the form 
\begin{align}
U(\bm{\theta}) &= \prod_{j=L} ^1 U_l(\bm{\theta}_l),\\
    U_l(\bm{\theta}_l)&=\prod_\nu e^{-\frac{i}{2}\theta_{l,A} A_\nu}\prod_\beta e^{-\frac{i}{2}\theta_{l,B} B_\beta}.
\end{align}
We model errors in the system by introducing single-qubit depolarizing noise for every rotation $e^{-\frac{i}{2}\theta_{l,A} A_\nu}$ on all qubits involved in the operation. 

The ground state of the toric code can be prepared by quasi-local Markovian jump operators \cite{dennis_topological_2002,verstraete_quantum_2009}. The total time required to produce the ground state is not reduced in the dissipative setting \cite{konig_Generating_2014}. 
However, in the circuit picture the quasi-local jumps can be implemented with a constant number of layers between resets. Therefore, we expect the ansatz to be robust against noise induced and unitary barren plateaus.

The dissipative ansatz we consider is inspired by circuits that map to ground states of the individual stabilizer terms \cite{barreiro_open-system_2011}. Usually, such a circuit is constructed from the following primitives: To transform to a +1 eigenstate of a $Z$ stabilizer $B_\beta = Z_{\beta,1}Z_{\beta,2}Z_{\beta,3}Z_{\beta,4}$ associated with vertex $\beta$, we first map the parity to an ancilla qubit using four CNOT operations with the involved qubits as control and the ancilla as a target. Then by applying a CNOT operation controlled by the ancilla, the parity is switched conditioned on the system parity being odd. This operation maps any state into an even parity state making it a +1 eigenstate of the vertex operator. 
To prepare ground states of the $X$ terms, the stabilizer readout needs to be sandwiched between Hadamard operations. We parameterize this ansatz by replacing all CNOT gates by $\text{C}R_X(\theta)$ gates with a $S$ operation on the target qubit which for $\theta = \pi$ perform a CNOT operation. Depolarizing noise is applied to all qubits involved in the stabilizer readout after every layer.
The depth of the layers that implement the jumps is independent of the number of physical qubits. As a consequence of the constant depth of the jumps, NIBPs do not present an issue.

\begin{figure}
    \centering
    \includegraphics{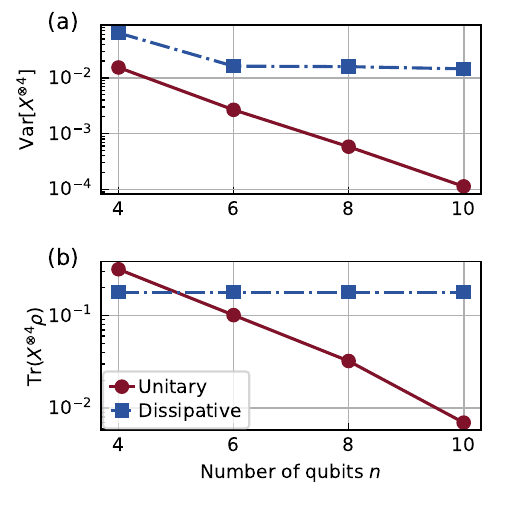}
    \caption{\textbf{Noise-induced barren plateaus.} The presence of noise leads to exponential concentration of expectation values in the depth of the circuit. In (a) we plot the variance of the gradient as a function of qubit number for the circuits that prepare toric code ground states. The variance expectation value of the unitary circuit is exponentially suppressed in system size, while that of the dissipative circuit plateaus. In the presence of noise the worst case expectation value is also exponentially concentrated. We show this in panel (b) where we plot the expectation value of one Hamiltonian term of the toric code Hamiltonian. The unitary circuit exponentially approaches a zero while the dissipative circuit does not.}
    \label{fig:Noise_varconcnetration_detconcentraion}
\end{figure}
In Fig.~\ref{fig:Noise_varconcnetration_detconcentraion} we show numerical evidence for the absence of noise-induced barren plateaus in dissipative learners. In panel (a), we plot the variance of the expectation value of a single Pauli string for both unitary and dissipative circuits across different system sizes.
The system we consider is defined on a rectangular lattice of height $2$ and width $n/2$. This allows us to simulate more system sizes before running into computational bottlenecks. In this setting, the depth of the unitary circuit has to increase as $n/2$ to prepare the ground state while the dissipative jumps remain constant depth. As a consequence, the expectation value of the unitary exponentially concentrates at the value of the fully mixed state. The deterministic concentration is shown in panel (b). Here, we evaluate the expectation value of a single Hamiltonian term for the trained parameters. The expectation value of the single term remains constant with system size for the dissipative ansatz but is suppressed in for the unitary one. The gradients exhibit the same behavior as shown in SI. \ref{app:toric_NIBP}. These results show that a dissipative ansatz can prepare toric code ground states in the presence of noise and, in turn, that there exist practical problems for which unitary learners suffer from NIBP while dissipative learners can avoid them.

\section{Discussion:}
We have demonstrated that parameterized dissipative algorithms can be scaled in the presence of noise. Specifically, we show that the addition of individual qubits that are periodically reset during the algorithm's execution enables the removal of entropy from the system, thereby mitigating noise-induced barren plateaus. Beyond this, the reset mechanism constrains the exploration of the ansatz to a restricted subspace, effectively preventing unitary barren plateaus. Through these mechanisms, dissipative quantum algorithms overcome some of the most pressing scalability challenges faced by variational quantum algorithms. As a result, they can prepare topologically ordered states, such as the toric code ground state, even in the presence of noise.

These insights not only highlight fundamental scalability advantages, but also open up opportunities to enhance existing dissipative protocols.
For example, in the context of ground-state preparation, the success of cooling algorithms, such as Refs.~\cite{raghunandan_initialization_2020, polla_Quantum_2021, mi_stable_2024}, crucially depends on the choice of bath splitting. Our work suggests that by parameterizing these dissipative circuits, one could optimize this splitting and improve both efficiency and accuracy of such cooling schemes.

Beyond algorithmic improvements, our findings also invite us to revisit the limits of dissipative quantum computing architectures. For instance, \citet{verstraete_quantum_2009} propose a continuous-time evolution governed by a Lindbladian designed to perform arbitrary quantum computations consisting of $T$ gates. The steady state of this evolution has an overlap of $1/(T+1)$ with the desired computational output. The Lindbladian dynamics involves both unital and non-unital operations. This raises the question: Is the construction robust to noise? The answer is negative---the scheme offers no greater noise robustness than standard unitary quantum computation. Achieving the desired output requires that all $T$ unital steps occur without error leading to an exponential suppression of the desired output. As we discuss in detail in SI.~\ref{app:VWC}, this highlights the importance of developing dissipative schemes that do not simply replicate unitary protocols in a dissipative setting but actively leverage dissipation to overcome noise-related limitations.

In the presence of unital, depolarizing noise the output state of any unitary quantum circuit exponentially approaches the fully mixed state in circuit depth allowing for classical simulation \cite{aharonov_limitations_1996,aharonov_polynomial-time_2023,fontana_classical_2023,gonzalez-garcia_pauli_2025}. The expectation values of local operators in random circuits can be classically simulated even in the presence of non-unital noise \cite{mele_noise-induced_2024,fefferman_effect_2024}. However, by carefully exploiting the dissipation, one can run circuits beyond logarithmic depth ensuring the non-simulatability using classical computers \cite{ben-or_quantum_2013,shtanko_complexity_2024}. This nuanced interplay between noise and dissipation is an exciting direction for future research.

Local error-correcting codes that act on only $\mathcal{O}(n)$ physical qubits, where $n$ denotes the number of logical qubits, inevitably lead to an exponential increase in entropy. Interestingly, for certain quantum states, dissipative algorithms are capable of circumventing this entropy growth. Identifying and characterizing the class of states for which dissipative dynamics can escape this exponential scaling remains an important and open question.

Looking ahead, we anticipate that dissipative quantum algorithms may find applications well beyond state preparation, motivating further research into their conceptual advantages. Certain computational tasks are inherently well suited to dissipative dynamics; for example, Gibbs states naturally emerge from interactions with an environment, and their mixed-state character necessitates nonunitary processes such as dissipation for their preparation \cite{zapusek_nonunitary_2023,van_mourik_experimental_2024,zapusek_variational_2025}. Beyond state preparation, dissipative frameworks may provide new avenues for learning error-correcting codes, where the task can be viewed as cooling into the ground-state manifold of a stabilizer Hamiltonian \cite{reiter_dissipative_2017,cong_quantum_2019,rojkov_stabilization_2024,rojkov_scalable_2025}. More broadly, dissipative circuits differ fundamentally from unitary ones in their ability to remove information, which could offer advantages in machine learning contexts by eliminating irrelevant features and improving generalization performance \cite{banchi_Generalization_2021,muser_provable_2024,hu_overcoming_2024}. Together, these perspectives suggest that dissipation is not just a tool for noise mitigation, but a powerful resource to expand the scope and capability of quantum algorithms.

\begin{acknowledgments}
We thank Wojciech Adamczyk, Vasilis Belis, and Jarrod R. McClean for helpful discussions and feedback. F.R., I.R., and E.Z. acknowledge funding from the Swiss National Science Foundation (Ambizione grant no. PZ00P2 186040) and the ETH Research Grant ETH-49 20-2. 
The project was conceived by E.Z. and F.R. The theoretical results were derived by E.Z. Numerical simulations were carried out by E.Z. in collaboration with I.R. F.R. provided guidance throughout the project. All authors contributed to the discussion, interpretation of the results, and writing of the manuscript.
\end{acknowledgments}


\section{Methods:}
\subsection{Preliminaries}
Throughout our calculations, it is useful to expand density matrices in the Pauli tensor product basis. For a single-qubit state, we write
\begin{align}
    \rho = \frac{1}{2}\left(I + \sum_{j=1}^{3}v_j F_j \right) = \frac{1}{2}(I + \bm{v} \bm{F}),
\end{align}
with $\bm{F} = (X,Y,Z)$ the vector of Pauli operators. We refer to $\bm{v}$ as the \textit{coherence vector}. Any unital noise channel transforms it linearly; $\bm{v} \rightarrow \bm{M} \bm{v}$ \cite{singkanipa_beyond_2025}. Using the left and right invariance of the Haar measure, we can transform the channel into the so-called \textit{ normal form} where the matrix $M$ takes the simple form:
\begin{align}
    \bm{M} = \begin{pmatrix}
         D_X  &0  & 0 \\
         0 & D_Y & 0 \\
         0  & 0 & D_Z 
    \end{pmatrix}.
\end{align}
See Refs. \cite{king_minimal_2001,verstraete_quantum_2002} and Lemma~\ref{lem:UnitalChannelProp}. The assumption that $\mathcal{N}$ is not the fully depolarizing channel guarantees that one $D_Q$ is non-zero. 
The expectation value of a traceless observable $O = \bm{a} \cdot \bm{F}$ can be written as $\Tr (O \rho )=  \bm{a}\cdot \bm{v}$.
With respect to this scalar product, we can define the adjoint of a quantum channel $\mathcal{N}$ as the map $\mathcal{N}^*$ satisfying
\begin{align}
    \Tr (O \mathcal{N}(\rho )) =  \bm{a}\cdot \bm{M} \bm{v} = \bm{M}^T \bm{a}\cdot \bm{v} = \Tr (\mathcal{N}^*(O) \rho ).
\end{align}

The Haar measure defines the uniform distribution over the set of unitary matrices. The most important properties of the Haar measure that we use are introduced in SI.~\ref{SI:Haar}, for further reading, consider Refs. \cite{collins_integration_2006,christandl_structure_2006,dankert_exact_2009,mele_introduction_2023}.
 An essential ingredient in our proofs is the fact that the Clifford group forms a two-design, meaning that it matches the Haar measure up to the second moment \cite{gross_evenly_2007,webb_clifford_2016}. This property allows us to replace Haar averages of second-moment expressions with averages over the finite Clifford group. In particular, for a real function $g:U(2) \to \mathbb{R}$ we can lower bound the Haar expectation value
\begin{align}
    \underset{U \sim U(2)}{\mathbb{E}}\left[ g(U)^2\right] 
    &= \frac{1}{\abs{\operatorname{CL}(2)}}\sum_{U \in \operatorname{CL}(2)} g(U)^2  \nonumber \\
    &\geq \frac{1}{\abs{\operatorname{CL}(2)}} g(C_\text{fixed})^2. \label{eq:Lower_Clifford_fixed_methods2}
\end{align}
Here, $C_\text{fixed}\in \operatorname{CL}(2)$ denotes an arbitrary fixed Clifford element.

\subsection{Proof of lower bound of the variance} \label{sec:VarLowerMethods}
Here we provide a proof of Proposition~\ref{prop:variance_lowerbound_main}. More details are presented in SI.~\ref{app:Variance_Lower_app}.

\begin{proof}[Proof of Proposition~\ref{prop:variance_lowerbound_main}]
We lower bound the variance after applying $M$ jumps. To begin, we consider the expectation value of the observable and use the fact that the final layer of single-qubit gates forms a global one-design \cite{mele_noise-induced_2024}:
\begin{align}
    \mathbb{E}\left[ \Tr \left( O\, \Phi^M\!(\rho_0)\right)\right] &=  
    \mathbb{E}\left[ \Tr \left( \mathbb{E}\left[ V^\dag O  V\right]\Phi'^M(\rho_0)\right)\right] \label{eq:ExpectationVanishes}\\
   &= \mathbb{E}\left[ \Tr \left( \frac{a_{I^{\otimes n}}}{2^n} I \Phi'^M(\rho_0)\right)\right]  = a_{I^{\otimes n}}, \nonumber
\end{align}
where $\Phi'$ denotes the circuit with the final layer of single-qubit gates removed. We evaluate the variance $
\operatorname{Var} \left[ \Tr ( O\, \Phi^M\!(\rho_0)) \right] =\mathbb{E} \left[ \Tr \left( O\, \Phi^M\!(\rho_0)\right)^2 \right] -\mathbb{E}\left[ \Tr \left( O\Phi^M(\rho_0)\right)\right]^2
$.
By applying the first point of Lemma~\ref{lem:SecondMomentSingle} (see also \cite[Lemma 16]{mele_noise-induced_2024}) to the last layer of gates, we can rewrite 
\begin{align}
\mathbb{E}&\big[\Tr(O\, \Phi^M\!(\rho_0))^2\big] =\!\!\!\!\!\! \sum_{P \in \{I, X, Y, Z\}^{\otimes n}} a_P^2 \, \mathbb{E}\big[\Tr(P \Phi'^M(\rho_0))^2\big] \nonumber\\
&= a_{I^{\otimes n}}^2 + \!\!\!\!\!\!\!\!\!\!\!\!\sum_{P \in \{I, X, Y, Z\}^{\otimes n} \setminus \{I^{\otimes n}\}} a_P^2 \, \mathbb{E}\big[\Tr(P \Phi'^M(\rho_0))^2\big]. \label{eq:Varaince_writtenout_methods}
\end{align}
The squared expectation values of the Pauli strings can be treated separately. The unital noise channel $\mathcal{N}$ is sandwiched between Haar random unitaries, allowing us to use their unitary invariance to bring the channel to normal form. The action of the adjoint noise channel $\mathcal{N}^*$ on a Pauli string is
\begin{align}
    \mathcal{N}^*(Q) = D_Q Q.
\end{align}
Since $\mathcal{N}$ is not fully depolarizing, the maximal factor satisfies $D_{\max} = \max_{Q \in \{X,Y,Z\}}(D_Q) > 0$. 
Moreover, for any nontrivial Pauli string $P$, there exists a Clifford unitary that maps the individual non-identity components of $P$ to $Q$. We can then lower bound
 \begin{align}
    & \quad \quad\quad \quad\quad \quad\quad \quad \mathbb{E}\big[\Tr(P \Phi'^M(\rho_0))^2\big] \nonumber\\
     &\geq \frac{1}{\abs{\operatorname{CL}(1)}^{2 \vert P \vert}} \mathbb{E}\left[ \Tr (\mathcal{N}^{*\otimes n}(Q^{\otimes \operatorname{supp}(P)}_\text{max} )  \Phi''_{\left[LM,1 \right]})^2\right] \nonumber\\
         & = \left(\frac{D_{\operatorname{max}}}{\abs{\operatorname{CL}(1)}}  \right)^{2  \vert P \vert} \mathbb{E}\left[\Tr( Q^{\otimes \operatorname{supp}(P)}_\text{max}\Phi''_{\left[LM,1 \right]})^2 \right]. \label{eq:methods_fixedshrinking}
 \end{align}
\begin{figure}
    \centering
    \includegraphics{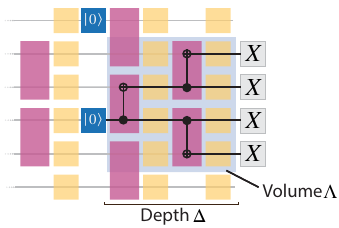}
    \caption{\textbf{Operator transformation example} A circuit illustrating how two-qubit Clifford gates can be fixed to transform observables. Specifically, the circuit maps the observable $X^{\otimes 4}$ to $I \otimes I \otimes X \otimes I$. The required gate depth $\Delta$ and the volume $\Lambda$ of fixed gates enabling this transformation are highlighted.}
    \label{fig:Variance_examplecircuit_methods}
\end{figure}
We now use the layers of the quantum circuit to transform the $\abs{P}$-local Pauli operator into a single-qubit Pauli operator on a qubit that is reset. We illustrate this in one dimension for Pauli string $X^K$ in Fig.~\ref{fig:Variance_examplecircuit_methods}. We apply: $\left[I^{\otimes K-1}  \otimes \text{CNOT}^\dag \right] X^{\otimes K} \left[  I^{ \otimes K-1}  \otimes \text{CNOT} \right] = X^{\otimes K-1} \otimes I $, shrinking the support of the observable by one in every layer. 
Each layer also applies a noise channel, contributing a suppression factor. We can upper bound the number of qubits in the support of the observable by the volume of the $d$-dimensional hypercube with edge length $K$ and therefore $D_{\max}^{2\abs{P} } \geq D_{\max}^{2 K^d }$.

Furthermore, because of the fixing gates, we accumulate a factor $C^{-1}$ with $C \leq \abs{\operatorname{CL}(2)}\abs{\operatorname{CL}(1)}$ for every qubit and layer involved in the transformation. To calculate the total factor, we need to understand what volume of gates needs to be fixed to transform the observable to a single-qubit observable on a reset qubit.

For this transformation, we need to fix the 2 qubit gates in the path. There are two cases to distinguish:
\begin{enumerate}
    \item No qubit in the path of the Pauli is reset. 
    \item One or multiple qubits are reset. 
\end{enumerate}
In the first case, no qubit in $\operatorname{supp}(P)$ is reset, and we must establish a connection to the nearest reset qubit. We choose an arbitrary direction $\hat{e}_i$. In the worst case, the operator is positioned symmetrically between two reset qubits in this direction. Consequently, the edge of the operator is at most 
\begin{align}
    \delta = \frac{1}{2}\left(\frac{n}{n_r}\right)^{1/d} \!\!+ \,\frac{K}{2}
\end{align}
away from the nearest reset qubit. In each unitary step, we can reduce the support of the operator by at least one site in that dimension, and we must perform this contraction in all $d$ dimensions. Therefore, the depth required to reduce the Pauli operator to a single-qubit operator is at most $d(\delta-1)$. In addition, we need to fix the gates until the next reset occurs, which can take up to $L$ layers. Thus, the total depth needed to contract the operator to a single site is at most $ d(\delta-1) + L.$
The volume of fixed gates can then be upper bounded by $\delta^d \left( d(\delta-1) + L\right).$

In the second case, a qubit in $\operatorname{supp}(P)$ is reset. We can transform the operator to a single-qubit Pauli on this qubit. This requires a depth smaller than $d(K-1) + L$ and a volume smaller than $K^d(d(K-1) + L)$. When a qubit in the path is reset, we first apply a Clifford that maps the relevant Pauli to $Z$; after the reset, the operator becomes $\mathcal{A}(Z)^{\otimes 2 }\otimes R^{\otimes 2 } = (q I + (1-q)Z)^{\otimes 2 }\otimes R^{\otimes 2 }$ where $R$ is a Pauli string on the other qubits. Expanding the reset part yields four terms, and by the Pauli-mixing lemma~\ref{lem:Pauli2Mixing} only symmetric positive terms remain, we can drop all but the identity term, which now has a prefactor $q^2$.
 The maximum number of reset qubits in this volume is $ \Gamma = \frac{K^d}{n/n_r} \frac{d(K-1)+L}{L}$ resulting in a total factor $q^{2\Gamma}$.
 
We obtain the following upper bounds on the extent of gates we need to fix:
\begin{align}
    \Delta &\leq \max \left[ d (\delta-1)  + L,d (K-1)+L\right] , \nonumber\\
    \Lambda &\leq \max\left[   \delta^d \left( d (\delta-1)  + L\right),K^d(d (K-1) + L)\right],\label{eq:contractionparams_methods}\\
    \Gamma &\leq \frac{ n_r }{n} \frac{K^d(d(K-1)+L)}{L}. \nonumber
\end{align}
In the final layer before the application of the reset we choose Clifford gates that transform the operator to a $Z$ operator on a reset qubit. Let $LM -j$ denote the depth at which the contracted observable meets the reset qubit, satisfying $j \leq \Delta$, then combining the factors from fixing gates, noise and reset we find,
\begin{align}
    \mathbb{E}\big[\Tr(P \Phi'^M(\rho_0))^2\big] \geq \frac{1}{C^\Lambda} D_{\max}^{2 {K^d \Delta}} q^{2 \Gamma} \nonumber\\ \times \mathbb{E}\left[ \Tr\left( (Z \otimes I^{n-1})\mathcal{A}(\Phi_{\left[ L M - j,1 \right]}(\rho_0))\right)^2 \right]. \label{eq:Shrinking_methods}
\end{align}
We apply the adjoint map of the amplitude damping channel $\mathcal{A}$ to the Pauli string and use the first point of Lemma~\ref{lem:SecondMomentSingle} to retain only positive symmetric terms
\begin{align}
   \mathbb{E}&\left[ \Tr\left(\mathcal{A}^* (Z \otimes I^{n-1})\Phi_{\left[ L M - j,1 \right]}(\rho_0)\right)^2 \right]
 \nonumber\\
    &=\mathbb{E}\left[ \Tr\left( \left(q I^{n}+(1-q)Z\otimes I^{n-1}\right) \Phi_{\left[ L M - j \right]}(\rho_0)\right)^2 \right] \nonumber\\
    &\geq q^2 \mathbb{E}\left[ \Tr\left( I^{n}\Phi_{\left[ L M - j,1 \right]}(\rho_0)\right)^2\right] = q^2 \label{eq:Damping_methodsq2}
\end{align}
In the final equality we used that $\Phi_{\left[ L M - j ,1\right]}(\rho_0)$ is a physical state and therefore has trace 1. We combine Eq.~\eqref{eq:Damping_methodsq2} and Eq.~\eqref{eq:Shrinking_methods} to lower bound the variance of an individual term. By summing all the individual terms in Eq.~\eqref{eq:Varaince_writtenout_methods}, we obtain the lower bound:
\begin{align}
    \operatorname{Var} \left[ \Tr\left(O\, \Phi^M\!(\rho_0)\right)\right] \geq \sum a_P^2 \frac{1}{C^\Lambda} D_{\max}^{2 {K^d \,\Delta}} q^{2 \Gamma + 2}.
\end{align}
Collecting terms in the constant $A$ one can rewrite the lower bound as in Eq.~\eqref{eq:lower_bound_variance_main}.
The lower bound is independent of the number of jumps, $M$, allowing us to choose an arbitrarily large $M$. 
\end{proof}
\subsection{Proof absence of barren plateaus}\label{sec:AbsenceBarrenSketch}
Corollary~\ref{cor:AbsenceBarrenLogmain} stated in Sec.~\ref{sec:Analytic_results} is a corollary of more general Theorem ~\ref{thm:lower_derivative} that we present in SI~\ref{app:Barren_Lower_app}. The reader can find further details about the theorem and its proof in SI~\ref{app:Barren_Lower_app}. Here we provide a proof of the theorem from which the Corollary directly follows.
\begin{proof} [Proof of Corollary~\ref{cor:AbsenceBarrenLogmain}]
The gradient of the cost function with respect to $\theta_\mu$ evaluates to
\begin{align}
    &\partial_\mu C = \Tr\left(  \Phi_{[LM-i-1,1]}(\rho_0) \partial_\mu\Phi_{[LM,LM-i]} ^*( O) \right) \nonumber\\
    &= \Tr\left(  \Phi_{[LM-i-1,1]}(\rho_0) \partial_\mu e^{i H_\mu \theta_\mu}\Tilde{\Phi}_{[LM,LM-i]} ^*( O) e^{-i H_\mu \theta_\mu}\right)\nonumber\\
    &= i \Tr\left(  \Phi_{[LM-i-1,1]}(\rho_0) \left[ H_\mu,\Phi_{[LM,LM-i]} ^*( O)\right] \right).
\end{align}
 Analogous to Eq.~\eqref{eq:ExpectationVanishes} one can now show that the expectation value of the gradient vanishes (see Lemma~\ref{lem:ExpGradZero}). The variance now takes the simple form
\begin{align}
    \operatorname{Var} \left[ \partial_\mu C\right] = \mathbb{E}\left[\left( \partial_\mu C\right)^2\right] = \mathbb{E}\left[\left( f_0(O) \right)^2\right],
\end{align}
where we define 
\begin{align}
f_j(\cdot) := i \, \Tr \big( \Phi_{[1,k]}(\rho_0) \big[ H_\mu, \Phi^*_{[k,L-j]}(\cdot) \big] \big).
\end{align}
Furthermore, we define $f'_j(\cdot)$ analogously to $f_j(\cdot)$, but using $\Phi'_{[k,L-j]}$, which is the map $\Phi_{[k,L-j]}$ with the single-qubit unitaries removed.
We use the first point of Lemma \ref{lem:PauliMixingGradients} to re-express:
\begin{align}
    \mathbb{E}\left[\left( f_0(O) \right)^2\right] = \sum_{P \in \{I, X, Y, Z\}^{\otimes n}} a_P^2 \mathbb{E}\left[\left( f_0(P) \right)^2\right].
\end{align}
In the following step we track how these Pauli observables $P$ are affected by the noise model. We use the fact that the Clifford group forms a 2-design and, just as in Eq. \eqref{eq:methods_fixedshrinking}, we choose Clifford operations that map $P$ to $Q_{\max}$, the Pauli operator corresponding to the least shrunk direction $D_{\max}$. Then
\begin{align}
    \mathbb{E}\left[\left( f_0(P) \right)^2\right] \geq\left(\frac{ D_\text{max}^2 }{\operatorname{CL}(1)}\right)^{\abs{P}} \mathbb{E}\left[ \left( f_0''(Q_{\text{max}}^{\otimes\operatorname{supp}(P)}) \right)^2\right].
\end{align}
Here we define $f''$ using the map $\Phi''^*_{[k,L-j]} $, which excludes both noise and single-qubit unitaries.

Now, as in Eq.~\eqref{eq:Lower_Clifford_fixed_methods2}, we proceed fixing Clifford gates to transform the operator $Q_{\text{max}}^{\otimes\operatorname{supp}(P)}$ into an operator with support on $H_\mu$. We refer to this backward-evolved operator $Q_b$. In this path, we make sure that the diameter of the Pauli operator does not increase. Since $H_\mu$ is in the backward light cone of the observable $O$ we are guaranteed that such a path exists. 
This path passes through $i -1$ layers and requires us to fix less than $K^d (i -1)$ gates. Denoting the number of reset qubits in the path $\gamma = \frac{i-1}{L} K^d \frac{n_r}{n}$ we can lower bound
\begin{align}
\label{eq:backwardstogate_methods}
    \mathbb{E}\left[(f_0(P))^2\right] \geq \left(\frac{D_{\max}^2 }{C} \right) ^{K^d (i -1)} (1-q) ^{2 \gamma} \nonumber\\ \times \mathbb{E}\left[ \Tr \left( \Phi_{[LM - (i-1) ,1]}(\rho_0) i \left[ H_\mu , Q_{b}\right]\right)^2 \right].
\end{align}
The factor $C \leq \abs{\operatorname{CL}(2)}\abs{\operatorname{CL}(1)}$ accounts for fixing single and two-qubit gates and $(1-q)^{2 \gamma}$ accounts for the reset. To ensure that the operator $Q_b$ has nontrivial commutator it cannot be the identity, therefore we retain the $Z$ Pauli string after reset, giving a factor $(1-q)$.
We can expand the generator of the rotation in the two-qubit Pauli basis $H_\mu = \sum_{R \in \{I, X, Y, Z\} ^{\otimes 2}} b_{R} R$. Substituting
\begin{align}
    \mathbb{E}\left[ \Tr \left( \Phi_{[LM - (i-1) ,1]}(\rho_0) i \left[ H_\mu , Q_{b}\right]\right)^2 \right]  = \nonumber\\ \sum_{R \in \{I, X, Y, Z\}^{^{\otimes 2}}} b_R^2 \mathbb{E}\left[ \Tr \left( \Phi_{[LM - (i-1) ,1]}(\rho_0) Q_R\right)^2 \right],
\end{align}
where we defined $Q_R \coloneqq i \left[ R , Q_b\right]$. In the transformation, we made sure that the support of the operator does not grow and therefore $\operatorname{diam(Q_R) \leq K}$. Since $H_\mu$ is traceless and non-zero there exists a $b_{\tilde{R}} = \max_{R \in P(2)} [{b_R}]>0$. We choose the transformation such that the backward evolved $Q_b$ does not commute with $\tilde{R}$, then
\begin{align}
    \mathbb{E}\left[ \Tr \left( \Phi_{[LM - (i-1) ,1]}(\rho_0) i \left[ H_\mu , Q_b\right]\right)^2 \right] \nonumber\\\geq b_{\tilde{R}}^2 \mathbb{E}\left[ \Tr \left( \Phi_{[LM - (i-1) ,1]}(\rho_0) Q_{\tilde{R}}\right)^2 \right].
\end{align}
Since $b_{\tilde{R}}^2$ is the maximal coefficient in the Pauli basis we can lower bound it by the average coefficient
\begin{align}\label{eq:coeff_methods}
    b_{\tilde{R}}^2 \geq \frac{1}{16} \sum b_R^2 = \frac{1}{64} \norm{H_{\mu}}_2^2 \geq \frac{1}{64} \norm{H_\mu}_\infty^2.
\end{align}
Now we apply Proposition~\ref{prop:variance_lowerbound_main} to lower bound the variance of the observable $Q_{\tilde{R}}$ with $\operatorname{diam(Q_{\tilde{R}})} \leq K$ yielding the lower bound:
\begin{align} \label{eq:lowesinglegrad_methods}
   \mathbb{E}\left[ \Tr \left( \Phi_{[LM - (i-1) ,1]}(\rho_0) Q_{\tilde{R}}\right)^2 \right]  \geq \frac{1}{C^{\Lambda}} D_\text{max}^{2K^d\Delta}q^{2\Gamma+2}
\end{align}
with $\Lambda,\Delta$ and $\Gamma$ defined as in Eq.~\eqref{eq:contractionparams_methods}. 
Combining the results in Eq.~\eqref{eq:backwardstogate_methods}, Eq.~\eqref{eq:coeff_methods} and Eq.~\eqref{eq:lowesinglegrad_methods} we find that 
\begin{align}
    \operatorname{Var}\left[ \partial_\mu C \right] 
    &\geq  \sum_P a_P^2  \left(\frac{D_{\max}^2 }{C} \right) ^{K^d (i -1)} (1-q)^{2\gamma}\nonumber\\ 
    & \quad \times\mathbb{E}\left[ \Tr \left( \Phi_{[LM - (i-1) ,1]}(\rho_0) i \left[ H_\mu , Q_b\right]\right)^2 \right] \nonumber\\
    &\geq  \sum_P a_P^2 \left(\frac{D_{\max} ^2}{C} \right) ^{K^d (i -1)} \!\!\!\!\!\!\!\!\! \!\!\!\!\!\!\!(1-q)^{2\gamma}b_{\tilde{R}}^2\frac{1}{C^{\Lambda}} D_\text{max}^{2\Delta}q^{2\Gamma+2} \nonumber\\
    &\geq \sum_P a_P^2  \frac{D_{\max}^{2K^d (i -1) +2 \Delta}}{C ^{K^d (i -1)+\Lambda}}  \frac{\norm{H_{\mu}}_\infty^2}{64} q^{2\Gamma+2}(1-q)^{2\gamma}.
\end{align}
The assumptions of the proof ensure that $\Delta,\Gamma$ and $K$ scale at most logarithmically with the size of the system $n$ and, therefore, $\operatorname{Var}\left[\partial_\mu C\right] \geq \Omega(1/\operatorname{poly(n}))$.
\end{proof}
\bibliographystyle{apsrev4-2}
\bibliography{references_better_abrv}
\appendix
\begin{onecolumngrid}
\newpage


\section{Definitions and notation}
The following definitions and notational conventions are used throughout the supplemental material.
\begin{itemize}
\item For $n\in \mathbb{N}$ we use $\left[n\right]$ to refer to the set of integers $\{1,\dots,n\}$.
    \item $\mathcal{B}(\mathcal{H})$ denotes the space
of bounded linear operators acting on Hilbert space $\mathcal{H}$.

\item For any two $A,B \in \mathcal{B}(\mathcal{H})$ we denote the Hilbert-Schmidt inner product $\langle A, B \rangle = \Tr (A^\dag B).$

\item The \textit{singular values} of a matrix $A \in \mathcal{B}(\mathcal{H})$ are the square roots of the eigenvalues of $A^\dag A$. We denote them $\sigma_i(A)$ with $\sigma_1 \geq \sigma_2 \dots \geq \sigma_q$.

\item Let $|A| \coloneqq \sqrt{A^\dagger A}$. The Schatten $p$-norm of an operator $A \in \mathcal{B}(\mathcal{H})$ is defined for $1 \leq p < \infty$ as 
$
\|A\|_p \coloneqq \big(\Tr(|A|^p)\big)^{1/p} = \bigg(\sum_i \sigma_i^p\bigg)^{1/p}
$.
Here $\sigma_i$ are the singular values of $A$.

In particular, the trace norm $\|A\|_1 \coloneqq \Tr(|A|)$ represents the sum of the singular values. The Hilbert-Schmidt (or Frobenius) norm $\|A\|_2 \coloneqq \sqrt{\langle A, A \rangle} = \sqrt{\Tr(|A|^2)}$ corresponds to the square root of the sum of the squared singular values. The operator norm $\|A\|_\infty \coloneqq \sup_{\|v\| = 1} \|A v\|$ equals the largest singular value.

Unless explicitly stated otherwise, $\|A\|$ will refer to $\|A\|_\infty$. For vectors $v \in \mathcal{H}$, $\|v\|$ will denote the Euclidean norm (2-norm).

\item For matrix $A \in B(\mathcal{H})$ and $1 \leq p \leq q$, we have $ \|A\|_q \leq \|A\|_p $.

\item For $x$ a real number we use $\ceil{x} = \min\left[ m \in \mathbb{Z} : m \geq x \right]$ and $\floor{x} = \max\left[ m \in \mathbb{Z} : m \leq x \right]$ to denote the ceiling and floor functions.
\item $U(d)$ denotes the unitary group of dimension $d$. 

\item The $n$-qubit Pauli group, denoted as $ P(n) $, is the group of all tensor products of single-qubit Pauli matrices, including the identity matrix and phase factors $\{1, -1, i, -i\}$. It is formally defined as:
\[
P(n) = \{ \alpha P_1 \otimes P_2 \otimes \cdots \otimes P_n \mid P_i \in \{I, X, Y, Z\}, \alpha \in \{1, -1, i, -i\} \}.
\]

\item For a $P \in P(n)$ we define the support as the set of indices on which $P$ acts non trivially; $\operatorname{supp}(P) \coloneqq \{ i \in \left[n \right] \;\vert \; P_i \neq I\}$. Furthermore, we use the shorthand $\abs{P} = \abs{\operatorname{supp}(P)}$ to denote the number of nontrivial Pauli strings a operator consists of.

\item We define a \textit{d-dimensional lattice} as the set of point $\Lambda \coloneqq \{ \sum_{i = 1} ^d a_i \Hat{e}_i : a_i \in \mathbb{Z}\}$ for $\{\Hat{e}_i\}_{i=1}^d$ a basis of $\mathbb{R}^d$. When we refer to a lattice of $n$ sites of dimension $d$ we refer to $\Lambda = \{ \sum_{i = 1} ^d a_i \Hat{e}_i : a_i \in \mathbb{N} \text{ and } a_i \leq n^{1/d}\}$.
\item Consider a $d$ dimensional lattice $\Lambda$ we define the \textit{distance} between two points $x,y \in \Lambda$ as the Manhattan-metric $d(x,y) = \sum_{i=1}^d \abs{x_i - y_i}$.
\item For $O \in \mathcal{B}(\mathcal{H}) $ and $O = \sum_{P \in P(n)} a_P P$ we define the \textit{diameter} as:
\begin{align}
\operatorname{diam}(O) \coloneqq \max_{\substack{P \in P(n) \\ a_P \neq 0}} \left[ \max_{\substack{P_1, P_2 \in P \vert \\ P_1 \neq I,\, P_2 \neq I}} \left[d(P_1, P_2) \right]\right],
\end{align}
It captures the maximal distance between two nontrivial Pauli operators that make up $O$. 
\item We define the \textit{one-dimensional brickwork circuit} for $n$ evens as the circuit consisting independent random gates distributed according to $\mu$ for qubits $(1,2)(3,4),\dots,(n-1,n)$ followed by gates on $(2,3)(4,5),\dots,(n-2,n-1) $. For $n$ odd the circuit definition is similar.

\item We define the \textit{d-dimensional brickwork circuit} on a d-dimensional lattice as the concatenation of $d$ one-dimensional brickwork circuits arranged in direction of the basis $\{\Hat{e}_i\}_{i=1}^d$. 

\item The \textit{light cone} of an observable $O$ with respect to a quantum channel $\Phi$ is defined as $\operatorname{Light}(\Phi,O) \coloneqq \operatorname{supp}(\Phi^*(O)).$ Here $\Phi^*$ is the adjoint channel with respect to the scalar product. 

\item 
\textit{Chebyshev's inequality} states that the probability of a random variable deviating more than \(\delta\) from its mean is bounded by the ratio of the variance to \(\delta^2\). Let \(X\) be a random variable with expectation \(\mu = \mathbb{E}[X]\) and variance \(\sigma^2 = \operatorname{Var}[X]\). Then, for any real and positive \(\delta\), we have:
$
\operatorname{Pr} \left( |X - \mu| \geq \delta \right) \leq \frac{\sigma^2}{\delta^2}
$.

\item \textit{Asymptotic notation.} We use asymptotic notation to describe the growth of functions. A positive function $ f(n) $ is said to be $\mathcal{O}(g(n))$ if there exist positive constants $c_1$ and $n_0 \in \mathbb{N}$ such that $f(n) \leq c_1 g(n)$ for all $n \geq n_0$. Similarly, $f(n) = \Omega(g(n))$ if there exist positive constants $c_2$ and $n_0 \in \mathbb{N}$ such that $0 \leq c_2 g(n) \leq f(n)$ for all $n \geq n_0$. We write $f(n)= \Theta(g(n))$ if $f(n) = \mathcal{O}(g(n))$ and $f(n) = \Omega(g(n))$ \cite{cormen_introduction_2022}.

\end{itemize}
\section{Entropy of quantum states}
Current quantum computers are small and noisy. Without error correction, incoherent errors increase the entropy of the state they operate on. Here we compare the entropy of states prepared using various methods: Unprotected quantum states, states with partial error correction, and states prepared on fully error-corrected quantum computers.  The quantum circuit we consider for our comparison operates on $n$ qubits and consists of $L$ unitary layers, where $L = \Omega(n)$. We assume geometric locality, meaning that error correction is performed on a $d$-dimensional lattice with local connections on the lattice. Additionally, we impose a constraint on the number of physical qubits, requiring that it scales at most linearly with the problem size, i.e., $n_p = \mathcal{O}(n)$. The noise on the computer is modeled by depolarizing noise, applied between the unitary layers. 

We want to lower bound the entropy of a quantum state that is produced by a circuit initialized in a pure state consisting of $L$ layers of unitaries interleaved with depolarizing noise. First we consider how the entropy of a quantum state changes if the depolarizing channel is applied to a subset of qubits. We use the quantum version of Shearer's inequality.
\begin{lemma}\label{lem:Shearer}(Quantum Shearer's inequality).
Consider $t \in \mathbb{N}$ and a family $\mathcal{F} \subset 2^{\{ 1,\dots,n\}}$ of subsets of $\{1,\dots,n\}$ such that each $i \in \{1,\dots,n\}$ is included in exactly $t$ elements of $\mathcal{F}$. Then for any $\rho \in B( (\mathbb{C}^2) ^{\otimes n} )$
    \begin{align}
        S(\rho) \leq \frac{1}{t} \sum_{F\in \mathcal{F}} S(\rho \vert_F).
    \end{align}
\end{lemma}
Here we denote the reduced density matrix as $\rho \vert_F$ For proof of the Lemma, see Ref.~\cite{muller-hermes_relative_2016}. 

\begin{theorem} \label{th:entropy}
For any $\rho \in B(\mathcal{H})$, let the quantum channel $\mathcal{D}_p^{n_c} = \mathcal{N}_p^{\otimes (n - n_c)} \otimes I^{\otimes n_c}$ be a depolarizing channel that acts on the first $n - n_c$ qubits, while the identity channel $I$ acts on the remaining $n_c$ qubits. We have the following inequality for the von Neumann entropy:
    \begin{align}
        S(\mathcal{D}_p^{n_c} (\rho)) \geq (1-p) S( \rho) + p \log (2^{n-n_c}).
    \end{align}
\end{theorem}
A similar statement was 
proven in Ref.~\cite{aharonov_limitations_1996}. We present a slightly modified proof from Ref.~\cite{muller-hermes_relative_2016}. 

\begin{proof}
We denote $n-n_c = r$. The state after the application of the depolarizing channel can be expressed as 
\begin{align}
    \mathcal{D}_p^{n_c}(\rho) = \sum_{k=0}^{r} \sum_{F\in\mathcal{F}_k}p^k (1-p)^{r-k} \left( \bigotimes_{l \in F} I/2 \otimes \rho \vert_{F^c}\right)
\end{align}
using the concavity of the von Neumann entropy the state can be further simplified
\begin{align}
    S\left(\mathcal{D}_p^{n_c}(\rho)  \right) \geq \sum_{k=0}^{r} \sum_{F\in\mathcal{F}_k}p^k (1-p)^{r-k} (k \log(2) + S(\rho \vert_{F^c}) 
\end{align}
We now use the identity $\sum_{k=0}^{r} \binom{r}{k}p^k (1-p)^{r-k} k = p r$ and apply Shearer's inequality (Lemma~\ref{lem:Shearer}). Every $i$ appears in $t = \binom{r-1}{r-k-1} = \binom{r}{r-k} \frac{r-k}{r}$ elements of $F^c$. 
\begin{align}
    S\left( \mathcal{D}_p^{n_c} (\rho)  \right) \geq &p r \log(2) 
    + \sum_{k=0}^r \binom{r}{r-k} \frac{r-k}{r} (1-p)^k p^{n-k} S(\rho) \nonumber\\
    \geq & p (n-n_c) \log(2) + (1-p) S(\rho),
\end{align}
concluding the proof. 
\end{proof}

We can now use Theorem~\ref{th:entropy} to lower bound the entropy of the state after $L$ layers
\begin{align}
    S\left(\left(\mathcal{D}_p^{n_c} (\rho) \right)^L (\rho_0)\right)  \geq & (1-p)^L S(\rho_0) 
    + \sum_{k=0}^{L-1} (1-p)^k p \log(2^{(n-n_c)}).
\end{align}
The system is initialized in a pure state, making $S(\rho_0) = 0$. Using the geometric series
\begin{align}\label{Eq:entropy_lower}
     S\left(\left(\mathcal{D}_p^{n_c} (\rho)\right)^L (\rho_0)\right)  \geq  \left(1-(1-p)^L  \right) \log(2^{(n-n_c)}).
\end{align}
The state of the circuit exponentially approaches the fully mixed state of $n-n_c$ qubits. This bound is not particularly tight, but sufficient for our purpose. Through the use of the subadditivity of the entropy we neglect errors that propagate to the error-corrected qubits. This setting could be investigated more thoroughly using techniques from random circuit sampling \cite{dalzell_random_2021} or an approach similar to \cite{bultrini_battle_2023}. 
Eq.~\eqref{Eq:entropy_lower} enables us to lower bound the entropy of the non corrected setting ($n_c=0$).

Through error correction, the errors in a quantum system can be suppressed. We assume that the operations at our disposal are geometrically local and that we have $n_p = \mathcal{O}(n)$ physical qubits. The bound derived in Ref.~\cite{bravyi_tradeoffs_2010} relates the distance and number of logical qubits of a local stabilizer code to the number of physical qubits. Specifically, $k d^\alpha \leq c n_p $ for $k$ the number of logical qubits $d$ the distance, $n_p$ the number of physical qubits and constants $c$ and $\alpha$. Now for $k = n$ and $n_p = \mathcal{O}(n) $ the distance $d$ of the code must remain constant. We assume that logical errors can be described by a depolarizing channel and that the distance of the code remains constant. Using these assumptions we can directly apply the results of the non error corrected case replacing $p$ by $p_l < p.$
\begin{align} \label{Eq:Entropy_correction}
    S((\mathcal{N}_{p_l}^{\otimes n })^L (\rho_0))  \geq  \left(1-(1-p_l)^L  \right) \log(2^n).
\end{align}
In the partially error corrected setting, we correct a subset of qubits of constant size ($k = n_c = \text{constant}$) therefore, the distance of the code can increase with the size of the system $n$. To lower bound the entropy we can assume that we have $n_c = k$ qubits that are noiseless. While the remainder is affected by depolarizing noise. Furthermore, to lower bound the entropy we assume the errors do not propagate to the error-corrected qubits. We then use Eq.~\eqref{Eq:entropy_lower} to lower bound the entropy.

We consider a simple dissipative circuit that resets $n_r = \Omega(n)$ qubits. The total number of qubits $n_p = n_r + n = \mathcal{O}(n)$. In the final layer, we swap $n_r$ qubits with system qubits. Assuming the reset and swap occur without an error, the final state has entropy less than $S \leq \log(2^{n-n_r})$.

To better understand the regimes in which dissipative algorithms may offer advantages, we perform numerical simulations. In Fig.~\ref{fig:NumericEntropy}, we plot the entropy of the states for toric code ground state preparation across various system sizes and noise strengths. Both dissipative and unitary ans\"atze, discussed in SI.~\ref{app:toric_NIBP}, are trained for systems ranging from 4 to 10 qubits, with entropy calculated for the resulting states. To enable a meaningful comparison across different system sizes, we divide by the entropy of the fully mixed state. Our results show that the entropy of states prepared dissipatively remains approximately constant, whereas that of states prepared unitarily increases with system size. However, initially, the dissipatively prepared states exhibit a higher entropy. This can be attributed to the steady-state solution being a mixture of ground states of the toric code rather than a single ground state.
\begin{figure*}
    \centering
    \includegraphics{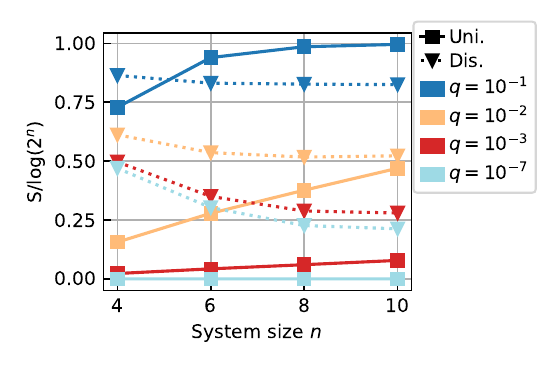} 
    \caption{\textbf{Numerical simulation of the entropy of noisy quantum circuits.} We numerically explore how the entropy of the output state scales compared to the fully mixed state for different system sizes $n$ and unitary and dissipative ans\"atze preparing toric code ground states. Unitary entropy is represented by solid lines and square markers while dissipatively prepared states are dotted with triangular markers. The entropy of the dissipative ansatz starts at a higher value as it natively prepares a mixture of the four degenerate ground states while the unitary ansatz prepares a single ground state. For growing system sizes the entropy of the unitarily prepared state approaches that of the fully mixed state while that of the dissipatively prepared state stays constant. 
    }
    \label{fig:NumericEntropy}
\end{figure*}

In Fig.~\ref{fig:EntropySize}, we extrapolate these trends, justified by the lower bounds in Eq.~\eqref{Eq:entropy_lower} and Eq.~\eqref{Eq:Entropy_correction}, to predict behavior at larger system sizes.
\begin{figure}[ht]
    \centering
    \includegraphics{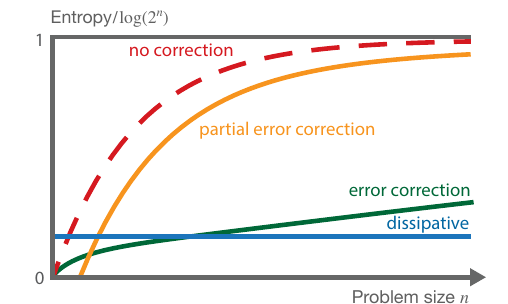}
    \caption{\textbf{Entropy of noisy quantum circuits.}
    In quantum algorithms subject to depolarizing noise, circuit depth scaling with system size leads to exponentially increasing entropy, quickly making computations unusable. Error correction is the solution, but it requires significant qubit overhead. If one only has $\mathcal{O}(n)$ physical qubits and local connectivity the distance of the code mus remain constant leading to an increase in entropy. Alternatively, one can encode a constant sized subset of the computational qubits with an increasing distance, this partial error correction still results in rising entropy. Dissipative algorithms offer a promising alternative, reducing entropy without the heavy resource demands of error correction.
    }
    \label{fig:EntropySize}
\end{figure}

\section{Quantum computation driven by dissipation in the presence of noise} \label{app:VWC}
Dissipative quantum algorithms can be scaled in settings where purely unitary algorithms cannot scale.
The authors of Ref.~\cite{verstraete_quantum_2009} show that purely dissipative dynamics can be used to perform universal quantum computations. How does this dissipative framework perform in the presence of noise? Here, we first introduce the construction of Ref.~\cite{verstraete_quantum_2009}. For a simple classical algorithm we show that the construction is not robust to noise. Specifically, we show that the overlap with the desired output of the computation is exponentially suppressed in the number of algorithmic steps.

Consider $N$ qubits on a line and nearest-neighbor unitary operators $\{U_t\}_{t=1}^T$. The state after the application of the $t$-th unitary is $\ket{\psi_t} = U_t U_{t-1} \dots U_1\ket{0^{\otimes N}}$. The evolution is described by a purely dissipative Liouvillian in Lindblad form 
\begin{equation}
    \dot{\rho} = \mathcal{L}(\rho)= \sum_k L_k \rho L_k^\dag - \frac{1}{2}\{ L_k^\dag L_k\,\rho \}. 
\end{equation}
The authors of Ref.~\cite{verstraete_quantum_2009} design geometrically local jump operators that have a steady state $\rho$ from which the desired computational output can be extracted. Specifically, overlap with the desired computational outcome is linearly suppressed in the number of gates $T$. The local jump operators are defined as
\begin{align}\label{eq:VWC_jumps}
    L_i &= \ketbraind{0}{i}{1} \otimes \ketbra{0}{0}, \\
    L_t &= U_t \otimes \ketbra{t+1}{t} + U_t^\dag \ketbra{t}{t+1},
\end{align}
here $i = 1,\dots, N$ and $t=1,\dots, T$. The steady state of the system is a uniform mixture over the computation history:
\begin{align}
    \rho = \frac{1}{T+1} \sum_t \ketbra{\psi_t}{\psi_t} \otimes \ketbra{t}{t}.
\end{align}
The desired output of the computation can be identified by the time register being in state $T$ and can be readout with probability $1/T$. It can be shown that the gap of the system does not depend on the system size and only polynomially depends on the number of gates $T$ guaranteeing fast convergence to the steady state. 

We show for a simple classical problem in the presence of noise the overlap of the steady state with the desired output state decays exponentially with the number of steps. From this we can conclude that the construction of Ref.~\cite{verstraete_quantum_2009}, despite being dissipative, is not robust to noise. Consider $T$ qubits on a line and a set of unitaries $\{U_t\}_{t=0}^T$ defined as
\begin{align}
    U_1 &= X_0 ,\\
    U_t &= \operatorname{CNOT}_{i-1,i} \text{ for i }\neq 1.
\end{align}
The initial unitary $U_1$ flips the zeroth qubit, preparing the state $\ket{\psi_1} = \ket{1 0^{\otimes {T-1}}}$. Subsequent $\text{CNOT}$ gates propagate this excitation along the chain, yielding intermediate states such as $\ket{\psi_2} = \ket{11 0^{\otimes {T-2}}}$ and eventually the desired output state $\ket{1^{\otimes{T}}}$ is reached.
We define the jump operators according to Eq.~\eqref{eq:VWC_jumps} on a computational register of $T$ qubits and a counter register with $T$ steps. We compare the noiseless ideal setting with a setting that includes incoherent flips of the computational register
\begin{equation}
    L_\text{noise}^i = \sqrt{\kappa} X_i. 
\end{equation}
All jumps map computational basis states to other computational basis states, therefore they do not build any coherences. Initializing in a computational basis state we can perform an entirely classical Markov chain simulation of the process. In Fig.~\ref{fig:VWC_noisy} we plot the overlap with the desired output state $\ket{1^{\otimes T}} $ as a function of system size $T$. In the ideal noiseless case, the overlap decays polynomially in the number of algorithmic steps $T$, as expected from the $1/(T+1)$ distribution over time steps. However, in the presence of noise, population spreads over the full $2^T$-dimensional Hilbert space, leading to an exponential suppression of the desired output. Since the overlap with the desired computational output is exponentially suppressed that algorithm cannot be scaled in the presence of noise. We conclude that the construction of Ref.~\cite{verstraete_quantum_2009} is not robust to noise.
\begin{figure}
    \centering
    \includegraphics{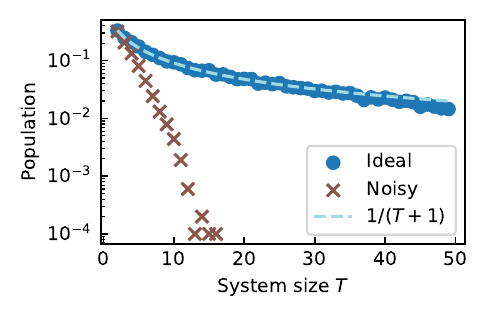}
    \caption{\textbf{Dissipative quantum computation.} In the ideal noiseless setting the overlap with the computational output state $\ket{\psi_T}$ scales as $1/(T+1)$. Noise spreads the population across $2^T$ states leading to the exponential suppression of the population in the desired output state.}
    \label{fig:VWC_noisy}
\end{figure}

\section{The Haar measure}\label{SI:Haar}
The Haar measure provides a way to define the uniform distribution over the set of unitary matrices \cite{collins_integration_2006,christandl_structure_2006,watrous2018theory}. Let $U(d)$ be the unitary group of degree $d=2^n$. 
The \textit{Haar measure} is the unique normalized and left and right invariant measure. For any integrable function $f: U(d) \mapsto \mathbb{C} $ the following holds:
\begin{align}
    \int_{U(d)} d\mu (U) f(WU) &= \int_{U(d)}d\mu
    (U) f(UW) = \int_{U(d)} d\mu (V) f(V),
\end{align}
where $W \in U(d)$. Additionally, the measure is normalized
\begin{align}
    \int_{U(d)} d\mu (U)  = 1.
\end{align}
A useful fact is that the haar measure is invariant under complex conjugation
\begin{align}
    \int_{U(d)} d\mu (U) f(U) = \int_{U(d)} d\mu (U) f(U^\dag).
\end{align}
In the following, we will use 
\begin{align}
    \mathbb{E}[f(U)] = \int_{U(d)} d\mu(U) f(U) \, ,
\end{align}
to denote the integral over the Haar measure. When we consider distributions over unitaries other than the Haar measure, we will state it explicitly. 

We state a few useful identities for integrating over the haar measure ~\cite{collins_integration_2006}. For any $O \in \mathcal{B}(\mathbb{C}^{d})$ 
\begin{equation}\label{eq:identityHaar1}
     \mathbb{E}[ U O U^\dag ]= \frac{\Tr\left[O \right]}{d} I.
\end{equation}
For any $O \in \mathcal{B}(\mathbb{C}^{2 d})$ we have:
\begin{align}\label{eq:identityHaar2}
    \mathbb{E}[ U^{\otimes 2} O U^{\dagger \otimes 2} ] = \frac{\Tr[O] I + \mathbb{F} \Tr[O \mathbb{F}]}{d^2 - 1} - \frac{\Tr[O] \mathbb{F} +\Tr[O \mathbb{F}]  I }{d(d^2 - 1)}.
\end{align}
Here, $\mathbb{F}$ is the swap operator defined as $\mathbb{F} \coloneqq \sum_{i,j} \ketbra{i,j}{j,i}$. 

\begin{definition}(Unitary t-design \cite{renes_symmetric_2004,gross_evenly_2007})
    A probability distribution over unitaries $\mathcal{E}=\{(p_i,U_i)\}$ is referred to as a t-design if and only if
    \begin{align}
        \underset{U \sim \mathcal{E}}{\mathbb{E}}\left[U^{\otimes t } O  U^{\dag \otimes t}\right]=\sum_i p_i U_i^{\otimes t } O  U_i^{\dag \otimes t }  = \underset{U \sim \text{U}(d)}{\mathbb{E}}\left[U^{\otimes t } O  U^{\dag \otimes t}\right] .
    \end{align}
    For all $O\in \mathcal{B}(\mathcal{H})$.
\end{definition}
\begin{definition}(Clifford Group \cite{gottesman_heisenberg_1998,gottesman_surviving_2016})
    The Clifford group is the set of unitaries that leaves the Pauli group $P(n)$ invariant under conjugation:
    \begin{align}
        \operatorname{CL}(n) \coloneqq \{ U \in 2^n : U P U^\dag \in P(n) \;\forall P \in P(n)\}.
    \end{align}
    It is a finite group generated by $\{\text{H, CNOT, S}\}$ \cite{gottesman_surviving_2016}.
\end{definition}
\begin{lemma}(The Clifford group forms a 2-design \cite{webb_clifford_2016})\label{lem:Clifford2design}
    The uniform distribution over the Clifford group $\operatorname{CL}(n)$ forms a 2-design.
\end{lemma}

\begin{lemma}(Pauli 2-mixing (Ref.~\cite[Lemma 14]{mele_noise-induced_2024}))\label{lem:Pauli2Mixing}
    Let $d = 2^n$ for $n\in \mathbb{N}$. For $P_1,P_2$ elements of the n-qubit Pauli group $P(n)$and $\mathcal{E}$ at least a 2-design
    \begin{align}
        \underset{U \sim \mathcal{E}}{\mathbb{E}}\left[U^{\otimes 2} (P_1 \otimes P_2 ) U^{\dag \otimes 2} \right] 
=
&\begin{cases}
    I \otimes I & \text{if } P_1 = P_2 = I, \\
    \frac{1}{d^2 - 1} \sum_{P \in \text{P}(n)  \setminus \{I\}} P \otimes P & \text{if } P_1 = P_2 \neq I, \\
    0 & \text{if } P_1 \neq P_2.
\end{cases}
    \nonumber\end{align}
\end{lemma}
\begin{proof}
    For the case $P_1 = P_2  = I$ the Haar average trivially gives $I \otimes I$. In the remaining cases, we apply Eq.~\eqref{eq:identityHaar2}. Pauli operators are traceless and therefore $\Tr\left[ P_1 \otimes P_2  \right] = 0$. 
    For the remaining terms we use the swap trick $\Tr\left[ P_1 \otimes P_1  \mathbb{F} \right] = \Tr\left[ P_1  P_2  \right] $ Due to orthogonality of the Paulis this gives $d$ for $P_1 = P_2$ and $0$ if $P_1 \neq P_2$. Summing the terms, we find
    \begin{align}
        \underset{U \sim \mathcal{E}}{\mathbb{E}}\left[U^{\otimes 2} (P_1 \otimes P_1 ) U^{\dag \otimes 2} \right]  = \frac{1}{d^2-1} \Tr\left[ P_1 \otimes P_1  \mathbb{F}\right] \mathbb{F} -  \frac{1}{d (d^2-1)}\Tr\left[ P_1 \otimes P_1  \mathbb{F}\right] I = \frac{1}{d^2-1} \sum_{P \in \text{P}(n)  \setminus \{I\}} P \otimes P.
    \end{align}
     This concludes the proof.
\end{proof}

\begin{lemma} \label{Lem:SingleGlobalOne}(A Layer of Single-Qubit Haar Random Gates Forms a Global 1-Design (Ref.~\cite[Lemma 15]{mele_noise-induced_2024}))
Let $\mathcal{E}$ be a distribution over the tensor product of single-qubit 1-design gates, where each unitary $U$ has the form $U = \bigotimes_{i=1}^n u_i$, with $u_i$ acting on the $i$-th qubit. Then, $\nu$ forms an $n$-qubit 1-design.
\end{lemma}
\begin{proof} Let $O \in \mathcal{B}(\mathbb{C}^d)$ with $d = 2^n$. Expanding $O$ in the Pauli basis, we have  
\begin{align}
\underset{U \sim \mathcal{E}}{\mathbb{E}}[U O U^\dagger] = \frac{1}{d} \sum_{P \in \{I, X, Y, Z\}^{\otimes n}} \Tr(OP) \underset{U \sim \mathcal{E}}{\mathbb{E}}[U P U^\dagger] = \frac{\Tr(O)}{d} I^{\otimes n},
\end{align}
where the last equality follows from the first-moment formula (Eq.~\eqref{eq:identityHaar1}) applied to each qubit and the traceless property of the Pauli matrices.
\end{proof}

\begin{lemma}\label{lem:SecondMomentSingle}
(From Ref.~\cite{mele_noise-induced_2024}: Second Moments of Single-Qubit Random Gate Layers)
Let $\mathcal{E}$ be a distribution over the tensor product of single-qubit 2-design gates, where $U = \bigotimes_{i=1}^n u_i$. For an operator $B \in \mathcal{B}(\mathbb{C}^d)$, the following hold: 
\begin{enumerate}
    \item Let $O = \sum_{P \in \{I, X, Y, Z\}^{\otimes n}} a_P P$, where $a_P \in \mathbb{R}$ for all $P$. Then, 
   \begin{align}
   \underset{U \sim \mathcal{E}}{\mathbb{E}} \big[ \Tr(O U B U^\dagger)^2 \big] = \sum_{P \in \{I, X, Y, Z\}^{\otimes n}} a_P^2 \underset{U \sim \mathcal{E}}{\mathbb{E}} \big[ \Tr(P U B U^\dagger)^2 \big]. 
   \end{align}
\item For any $P \in \{I, X, Y, Z\}^{\otimes n}$, 
   \begin{align}
   \underset{U \sim \mathcal{E}}{\mathbb{E}} \big[ \Tr(P U B U^\dagger)^2 \big] = \frac{1}{3^{|P|}} \sum_{\substack{\operatorname{supp}(Q) = \operatorname{supp}(P) \\ Q \in \{I, X, Y, Z\}^{\otimes n}} } \Tr(QB)^2.
   \end{align}
\end{enumerate}
\end{lemma}
\begin{proof} 
We aim to calculate $\underset{U \sim \mathcal{E}}{\mathbb{E}}[\Tr(O U B U^\dagger)^2]$. First, we rewrite the expression as
\begin{align}
\underset{U \sim \mathcal{E}}{\mathbb{E}}[\Tr(O U B U^\dagger)^2] = \underset{U \sim \mathcal{E}}{\mathbb{E}}[\Tr(O^{\otimes 2} U^{\otimes 2} B^{\otimes 2} (U^\dagger)^{\otimes 2})].
\end{align}
Expanding $O$ in the Pauli basis as $O = \sum_{P \in \{I, X, Y, Z\}^{\otimes n}} a_P P$, this becomes
\begin{align}
\sum_{P, Q \in \{I, X, Y, Z\}^{\otimes n}} a_P a_Q \underset{U \sim \mathcal{E}}{\mathbb{E}}[\Tr((P \otimes Q) U^{\otimes 2} B^{\otimes 2} (U^\dagger)^{\otimes 2})].
\end{align}
Using the cyclic property of the trace, we rewrite this as
\begin{align}
\sum_{P, Q \in \{I, X, Y, Z\}^{\otimes n}} a_P a_Q \underset{U \sim \mathcal{E}}{\mathbb{E}}[\Tr((U^\dagger)^{\otimes 2} (P \otimes Q) U^{\otimes 2} B^{\otimes 2})].
\end{align}
We now use that the Haar measure is invariant under complex conjugation; $\mathbb{E}_{U \sim \mathcal{E} }[f(U)] = \mathbb{E}_{U \sim \mathcal{E}}[f(U^\dag)] $ and the Pauli-mixing lemma \ref{lem:Pauli2Mixing}, we conclude that:
$
\underset{u_i \sim U(2)}{\mathbb{E}}\left[u_i^{\otimes 2} (P_1 \otimes P_2) (u_i^\dagger)^{\otimes 2}\right] = 0,
$
for two different single-qubit Pauli operators \(P_1\) and \(P_2\). The term simplifies to
\begin{align}
\sum_{P \in \{I, X, Y, Z\}^{\otimes n}} a_P^2 \underset{U \sim \mathcal{E}}{\mathbb{E}}[\Tr(P U B U^\dagger)^2],
\end{align}
proving the first statement of the Lemma.

To prove the second statement we express $P$ as a product of single-qubit Pauli matrices $P_i$ $P = P_1 \otimes P_2 \otimes \cdots \otimes P_n$. And factor $U = \bigotimes_{i=1}^n u_i$ as a tensor product of single-qubit unitaries. To conveniently write the expression define the permutation operator $\Pi$, which permutes the entries of the tensor product; $ \Pi^\dag (P_1 \otimes P_2)^{\otimes 2} \Pi =\Pi^\dag (P_1\otimes P_2 \otimes P_1 \otimes P_2) \Pi= P_1 \otimes P_1 \otimes P_2 \otimes P_2$. 
\begin{align}
    \underset{U \sim \mathcal{E}}{\mathbb{E}}[\Tr(P U B U^\dagger)^2] = \underset{U \sim \mathcal{E}}{\mathbb{E}}[\Tr(P^{\otimes 2} U^{\otimes 2} B^{\otimes 2} U^{\dagger \otimes 2})] = \Tr\left(\Pi \underset{U \sim \mathcal{E}}{\mathbb{E}}\left[\bigotimes_{i=1}^{n} u_i^{\otimes 2}(P_i \otimes P_i) \; \ u_i^{\dag\otimes 2} \right] \Pi^\dag B^{\otimes 2} \right)
\end{align}
To this expression we can apply Pauli 2-mixing (Lemma~\ref{lem:Pauli2Mixing}). 
\begin{align}
\underset{U \sim \mathcal{E}}{\mathbb{E}}[\Tr(P U B U^\dagger)^2] = \Tr \left( \left( \bigotimes_{i \in \operatorname{supp}(P)} \frac{1}{3} \sum_{Q_i \in \{X,Y,Z\}} Q_i \otimes Q_i\right) B^{\otimes 2} \right)\\
= \frac{1}{3^{\abs{P}}} \sum_{\substack{Q \in \{I,X,Y,Z\}^{\otimes n}\\ \operatorname{supp}(Q) = \operatorname{supp(P)}}} \Tr \left( Q^{\otimes 2} B^{\otimes 2}\right) \nonumber
.
\end{align}
\end{proof}

\section{Quantum channels and the coherence vector picture}
A \emph{quantum channel} is a completely positive, trace-preserving (CPTP) map $\Phi: \mathcal{B}(\mathcal{H}) \to \mathcal{B}(\mathcal{H})$. 
The map $\Phi$ satisfies complete positivity: For any integer $n \geq 1$, the extended map 
    \begin{align}
    \Phi \otimes I^{\otimes n}: \mathcal{B}(\mathcal{H} \otimes \mathbb{C}^n) \to \mathcal{B}(\mathcal{H} \otimes \mathbb{C}^n)
    \end{align}
is positive, where $I^{\otimes n}$ is the identity map on $\mathbb{C}^n$.
Furthermore, it preserves the trace of the density matrix: For all $X \in \mathcal{B}(\mathcal{H})$, 
\begin{align}
    \Tr[\Phi(X)] = \Tr[X].
\end{align}
The \emph{Kraus representation} of a quantum channel $\Phi$ is given by a set of at most $d^2$ operators $\{ K_i \}_{i=1}^{d^2}$ acting on $\mathcal{H}$ \cite{kraus_states_1983}. The action of the channel is 
\begin{align}
\Phi(\rho) = \sum_i K_i \rho K_i^\dagger,
\end{align}
where $\rho \in \mathcal{B}(\mathcal{H})$ is a density matrix and $K_i$ are the Kraus operators that satisfy the completeness relation
\begin{align}
\sum_i K_i^\dagger K_i = I.
\end{align}
The Kraus representation guarantees the complete positivity and trace-preserving properties of $\Phi$.
Given a quantum channel $\Phi$, we say that $\Phi$ is \emph{unital} if and only if it maps the identity operator to the identity operator, i.e., 
$\Phi(I) = I.$ Otherwise, we say that $\Phi$ is \emph{non-unital}.

\begin{lemma} (Unital channels can only decrease purity)\label{lem:PurityUnital}
    Given a unital quantum channel $\Phi$ and a density matrix $\rho \in \mathcal{B}(\mathcal{H})$ with purity $P = \Tr(\rho^2)$ then $P' = \Tr (\Phi(\rho)^2)\leq 1$. Equality holds iff $\Phi$ is unitary.
\end{lemma}
For proof, see Ref. \cite{singkanipa_beyond_2025}.

Given a quantum channel $\Phi: L(\mathbb{C}^d) \to L(\mathbb{C}^d)$, its adjoint map $\Phi^*: L(\mathbb{C}^d) \to L(\mathbb{C}^d)$ is defined as the linear map such that
\begin{align}
\langle \Phi^*(A), B \rangle = \langle A, \Phi(B) \rangle,
\end{align}
for any $A, B \in \mathcal{B}(\mathcal{H})$, where $\langle \cdot, \cdot \rangle$ denotes the Hilbert-Schmidt inner product. If $\{ K_i \}_{i=1}^{d^2}$ is a set of Kraus operators for $\Phi$, then the adjoint channel $\Phi^*$ can be expressed as
\begin{align}
\Phi^*(\cdot) = \sum_{i=1}^{d^2} K_i^\dagger (\cdot) K_i.
\end{align}
The adjoint channel $ \Phi^* $ is always unital, that is, $ \Phi^*(I) = I $, which follows from trace preservation. However, the adjoint is not necessarily trace-preserving: it holds that $ \Phi^* $ is trace-preserving if and only if $ \Phi$ is unital.
If the Kraus operators of the quantum channel $\Phi $ are Hermitian, then the adjoint channel coincides with the quantum channel, $\Phi^* = \Phi $.
If $ \Phi_1 $ and $ \Phi_2 $ are two quantum channels, then
\begin{align}
(\Phi_1 \circ \Phi_2)^* = \Phi_2^* \circ \Phi_1^*.
\end{align}
Moreover,
\begin{align}
(\Phi_1 \otimes \Phi_2)^* = \Phi_1^* \otimes \Phi_2^*.
\end{align}

\subsection{The coherence vector picture}\label{app:coherence_vec_pic}
On the space of bounded linear operators $B(\mathcal{H})$ we can choose a basis $\{ F_j\}_{j=0} ^{d^2-1}$ that satisfies $F_0 = I$ and $\Tr (F_j )  = 0 \;\; \forall j > 0$. Additionally, the basis vectors should be orthonormal with respect to the inner product.
\begin{align}
    F_j = F_j^\dag, \quad \langle F_k , F_j\rangle = \frac{1}{d} \Tr \left( F_k F_j\right) = \delta_{jk}  \quad \forall j,k.
\end{align}
The Pauli group, $P(n)$, is an example of such a basis for $n$ qubits.

With the help of this basis, we can represent any density matrix as:
\begin{align}
    \rho = \frac{1}{d}  \left( F_0 + \sum_{j=1} ^{d^2 -1} v_j F_j \right) = \frac{1}{d} \left( F_0 + \bm{v} \cdot \bm{F} \right),
\end{align}
where we introduced the vector notation $\bm{F} = \{F_1,\dots,F_{d^2-1}\}$, and refer to $\bm{v}=\{ v_1,\dots,v_{d^2-1}\},$ as the \textit{Bloch} or \textit{coherence vector.}
Any quantum channel $\Phi$ transforms the state to
\begin{align}
    \Phi(\rho) = \frac{1}{d} \left( F_0 + \bm{v'} \cdot \bm{F} \right).
\end{align}
The transformation of the coherence vector contains a linear and an affine part
\begin{align}
    \bm{v}' = M \bm{v} + \bm{c}.
\end{align}
Their entries can be calculated according to
\begin{align}
    M_{ij} &= \langle F_i,\Phi(F_j)\rangle, \label{eq:MsmallDef}\\
    c_j &= \langle F_j , \Phi(I)\rangle.\label{eq:cDef}
\end{align}
The purity of a state can be written as 
\begin{align} \label{eq:PurityCoherenceVector}
    P \coloneqq \Tr \left( \rho^2 \right) = \frac{1}{d} \left( 1 + \norm{v}^2\right).
\end{align}

 \begin{lemma} (Properties of unital channels) \label{lem:UnitalChannelProp}
 For unital channel $\Phi$ it holds:
 \begin{align}
     \bm{c} &= 0\label{eq:cisZero},\\
     \norm{\bm{v}'} &= \norm{M \bm{v}} \leq \norm{ \bm{v}}.
 \end{align}
equality holds iff $\Phi$ is unital. 
 \end{lemma} 
\begin{proof} To prove the first point observe $c_j = \langle F_j , \Phi( I ) \rangle = \langle F_j , I\rangle = 0$. 
 We now consider the second point. Using Eq.~\eqref{eq:PurityCoherenceVector} we can rewrite the norm of the coherence vector in terms of the purity
\begin{align}
    \norm{\bm{v'}}^2 = d P -1
\end{align}
Using Lemma~\ref{lem:PurityUnital} we know that $P' \leq P$ and consequently $\norm{\bm{v'}} = \norm{M \bm{v}} \leq \norm{\bm{v}}$.
From Lemma~\ref{lem:PurityUnital} we have $P' = P$ iff $\Phi$ is unitary. The condition $P' = P$ is equivalent to $\|M v\| = \|v'\| = \|v\|$, which therefore holds for all $v$ if and only if $\Phi$ is unitary.
\end{proof}

\begin{corollary}(Transfer matrix of unitary is orthogonal)
    For $U \in \mathcal{B}(\mathcal{H})$ a unitarUnity transformation. Let M be the corresponding transformation of the coherence vector. Then $M$ is orthogonal; $M^TM = I$.
\end{corollary}
\begin{proof} Due to Lemma~\ref{lem:UnitalChannelProp} we know that the channel leaves the norm of the coherence vector invariant. Consequently, it is orthogonal.
 \end{proof}
 
\begin{definition}[Contractive]\label{def:contractive}
A (finite-dimensional) map $\mathcal{N}$ is called HS-contractive if there exists an $r < 1$ such that for all states $\rho_1 \neq \rho_2$, the inequality $\|\mathcal{N}(\rho_1) - \mathcal{N}(\rho_2)\|_2 \leq r \| \rho_1 - \rho_2 \|_2$ holds.
\end{definition}
In the coherence vector picture the contractivity condition is equivalent to $\norm{M\bm{v}}\leq r \norm{\bm{v}}$ for $r <1$ \cite{singkanipa_beyond_2025}. The depolarizing channel is an example of such a contractive map.

\subsection{Single-qubit channels and normal form}
We consider a single-qubit system with $F_0 = I$ and $\bm{F} = \{X, Y, Z\}$. Then any density matrix can be represented as $\rho = 1/2(F_0 + \bm{v} \cdot \bm{F})$ with $\bm{v}\in \mathbb{R}^3$ and $\norm{\bm{v}}\leq 1 $.
We define the \textit{Pauli transfer matrix} as
\begin{align}
    \bm{M} = \begin{pmatrix}
1 & 0 \\
\bm{c} & M
\end{pmatrix}
\end{align}
where in accordance with Eq.~\eqref{eq:MsmallDef} and Eq.~\eqref{eq:cDef}
\begin{align}
    \bm{M}_{i,j} = \Tr\left( F_i, \Phi(F_j)\right).
\end{align}
The action of the channel on a Pauli matrix $F_j$
\begin{align}
    \Phi(F_j) = \sum 
    _{i = 0}^3\bm{M}_{i,j} F_i.
\end{align}
Using the Pauli transfer matrix we can find a concise representation of any channel. 
\begin{lemma}(Normal form)\label{lem:NormalForm}
    For any quantum channel $\Phi$ there exist $U,V \in U(2)$ such that for all $\rho \in \mathcal{B}(\mathcal{H})$
    \begin{align}
        \Phi(\rho) = U \Tilde{\Phi} (V \rho V^\dag) U^\dag.
    \end{align}
    And $\Tilde{\Phi}$ has Pauli transfer matrix of the form 
    \begin{align}
        \Tilde{\bm{M}} = \begin{pmatrix}
1 & 0 & 0 & 0 \\
c_X & D_X & 0 & 0 \\
c_Y & 0 & D_Y & 0 \\
c_Z & 0 &0 & D_Z
\end{pmatrix}
    \end{align}
    This is known as the \textit{normal form} of the channel. Furthermore, the parameters satisfy $\norm{\bm{c}}_2^2 \leq 1 $ and $D_P \leq 1.$
 \end{lemma}
\begin{proof}
For proof that unitaries exist to bring the channel in normal form see Ref.~\cite{king_minimal_2001,beth_ruskai_analysis_2002,mele_noise-induced_2024}. Let $\Phi$ be a channel in normal form with transfer matrix $M$ and $\rho = \frac{1}{2}(F_0 + \bm{v} \cdot \bm{F})$ be an arbitrary single-qubit state. Any physical state satisfies $\Tr(\rho^2)\leq 1$, this implies $\norm{\bm{v}}^2 \leq 1$. Since $\Phi$ is CPTP, it transforms $\rho$ into another valid quantum state with bounded purity: $\norm{M \bm{v} + \bm{c}}^2  = \sum_{Q \in \{X,Y,Z\}}(c_Q + D_Q v_Q)^2\leq 1 $. In particular, choosing $\bm{v} = \bm{0}$ gives $\norm{\bm{c}}^2 \leq 1$. Setting $\bm{v} = (\pm 1,0,0)$ in turn gives $(c_X \pm D_X)^2 \leq 1$ and thereby implies $c_X^2+ D_X^2 \leq 1$.
\end{proof}

\begin{lemma}\label{lem:TP_channel_evsmallerone}
Tensor products of single-qubit channels have eigenvalues smaller than one
\end{lemma}
    \begin{proof} Bringing the single-qubit channels to normal form is equivalent to conjugating the transfer matrix by orthogonal matrices corresponding to the unitary transformations $U$ and $V$. This operation leaves the magnitude of the eigenvalues invariant. Therefore, we can consider the eigenvalues of tensor products of quantum channels in normal form. A matrix in normal form is lower triangular with diagonal entries smaller than one (Lemma~\ref{lem:NormalForm}). For a lower triangular matrix the eigenvalues are simply the diagonal entries. The tensor product of lower triangular matrices is again lower triangular with the product of the diagonal entries on the diagonal. As all entries are smaller than one their product is also smaller than one. Consequently, the eigenvalues of tensor products of single-qubit channels are smaller than one.
\end{proof}
We will focus our attention to two quantum channels representing unital and non-unital noise. The archetypical unital channel is the depolarizing channel. It maps density matrix $\rho$ to
\begin{align}
    \mathcal{N}_p ( \rho ) = (1-p) \rho + p \frac{1}{2} I.
\end{align}
It drives all states uniformly towards the fully mixed state making the fully mixed state a stationary state and the channel unital. Visualized on the Bloch sphere, it can be interpreted as uniformly shrinking all directions of the Bloch sphere. It has normal form parameters $\bm{c} = (0,0,0)$ and $\bm{D}= (1-p,1-p,1-p)$. As $\bm{c}=0$ the depolarizing channel is a unital. 

The amplitude damping channel is the archetypical non-unital channel. It is defined via the Kraus operators:
\begin{align}
    K_1 = \begin{pmatrix}
        0 & \sqrt{q} \\
        0 & 0 
    \end{pmatrix},
    \quad 
    K_2 = \begin{pmatrix}
        1 &0 \\
        0 & \sqrt{1-q}
    \end{pmatrix}.
\end{align}
When applied to a density matrix, we obtain:
\begin{align}\label{eq:Ampdampdef}
    \operatorname{Ampdamp}_q(\rho) = K_1 \rho K_1^\dag + K_2 \rho  K_2 ^\dag  = \begin{pmatrix}
        \rho_{0,0} + q \rho_{1,1} & \sqrt{1-q}\rho_{0,1}  \\
         \sqrt{1-q}\rho_{1,0} & (1-q) \rho_{1,1}
    \end{pmatrix}.
\end{align}
For $q=1$ the channel maps all states to the pure state $\ket{0}.$ The representation in the normal form is $\bm{c} = (0,0,q)$ and $\bm{D} = (\sqrt{1-q},\sqrt{1-q},1-q)$. 
This operation is readily available on all platforms proposed for quantum computing as a quantum computer must be capable of initializing the system in a well-defined state \cite{sohn_topics_1997,kastler_Quelques_1950,wineland_Optical_1985,magnard_fast_2018}. Typically, this initialization is realized via the amplitude damping channel.
\section{Dissipative circuit model}\label{sec:DisCircuitModel}
We investigate quantum circuits consisting of unitary gates and unital and non-unital noise. The system evolution is captured by the channel:
\begin{align}
    \Phi^M = \Phi_{\left[ LM,1\right]},
\end{align}
where
\begin{align}
    \Phi_{\left[ a,b\right]} = \left(\mathcal{V}^a\circ  \mathcal{N}\circ \mathcal{U}^a \circ\mathcal{A}^{\chi(a)}\right) \circ \dots \circ \left(\mathcal{V}^b\circ  \mathcal{N}\circ \mathcal{U}^b \circ \mathcal{A}^{\chi(b)}\right).
\end{align}
Here, $\mathcal{V}^j(\cdot) = V^j (\cdot) V^{j\dag} $ denotes layers of single-qubit gates and each $\mathcal{U}^j = U^j (\cdot) U^{j\dag} $ consists of one layer of a $d$ dimensional brickwork circuit. 
$\mathcal{A}$ is a channel that resets $n_r$ equidistant qubits by applying an amplitude damping channel to them. Specifically, on a $d$ dimensional lattice the distance between reset qubits is less than $\ceil{ \left( \frac{n}{n_r}\right)^{1/d}}$ and more than $\floor{ \left( \frac{n}{n_r}\right)^{1/d}}$ in every direction. We denote the probability of the qubit being reset $q$ (It maps the fully mixed state of a single-qubit to a state with fidelity $(1+q)/2$ with state $\ket{0}$, 
see definition of the amplitude damping channel Eq.~\eqref{eq:Ampdampdef}). The function $\chi$ ensures that the amplitude damping channels are applied after $L$ layers
\begin{align}
\chi(l) = \begin{cases} 
1 & \text{if l } \bmod L = 0, \\
0 & \text{otherwise.}
\end{cases}
\end{align}
We refer to a block of $L$ layers as a \textit{jump}, ma rking the interval between consecutive reset operations.
We focus on the average-case behavior of these circuits, assuming they are sampled from a distribution that forms a unitary two-design---that is, one that matches the uniform (Haar) distribution over unitaries up to the second moment  \cite{renes_symmetric_2004,gross_evenly_2007}. The channel $\mathcal{N}$ models the noise. It is a $n$ qubit tensor product of unital single-qubit contractive channel that is distinct from the fully depolarizing channel. Specifically, there exists a $Q \in \{X,Y,Z\} :   D_Q \neq 0.$

\begin{figure}
    \centering
    \includegraphics{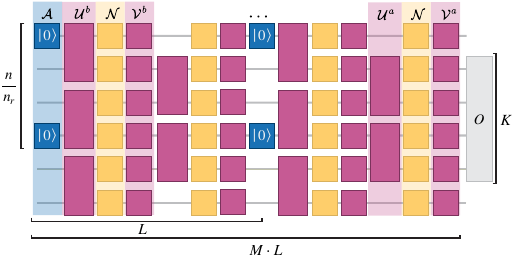}
    \caption{\textbf{Circuit model} The circuit model in one dimension. The circuit consists of channel $\mathcal{A}$ resetting equidistant qubits, two-qubit unitaries $\mathcal{U}$, noise $\mathcal{N}$ and single-qubit unitaries $\mathcal{V}$. We show $\Phi_{[a,b]}$. Here $\chi(b) = 1$, and $\chi(a) = 0$ therefore we apply a reset in the first shown layer but not in the last. The depth of a jump, referring to the number of layers between consecutive resets, is $L = 2$. The fraction of qubits that is reset is described by $n_r/n$ and $M$ captures the number of jumps that are applied. $O$ is the observable that is measured, it has diameter $K$.}
    \label{fig:circuit_model}
\end{figure}
Our proofs rely on peeling back the layers of $\Phi$ one by one. To conveniently express this we denote 
\begin{align}
    \Phi_{\left[ a,b\right]}' = \left( 
     \mathcal{N}\circ \mathcal{U}^a \circ\mathcal{A}^{\chi(a)}\right) \circ \dots \circ \left(V^b\circ  \mathcal{N}\circ \mathcal{U}^b \circ \mathcal{A}^{\chi(b)}\right).
\end{align}
as $\Phi_{\left[ a,b\right]}$ with the last layer of single-qubit gates $V^a_{\text{single}}$ removed. Similarly, $\Phi''_{\left[ a,b\right]}$, has both the single-qubit gates and noise $\mathcal{N}$ removed and so on. Finally, adding four dashes is equivalent to removing an entire layer: $\Phi''''_{\left[ a,b\right]} = \Phi_{\left[ a-1,b\right]}$.



\section{Variance lower bound}\label{app:Variance_Lower_app}
In this section, we establish a lower bound for the variance of observables measured after the application of a dissipative quantum circuit. 
\begin{proposition}(Lower bound of the variance) \label{prop:variance_lowerbound}
     Consider a circuit as described in Sec.~\ref{sec:DisCircuitModel} defined on a lattice of $n$ qubits and dimension $d$. Every $L$ layers, $n_{r}$ equidistant qubits are reset with probability $q$. Let $O \in \mathcal{B}(\mathcal{H}) $ be an observable with $\operatorname{diam}(O) = K$ and Pauli decomposition $O= \sum_{P} a_P P $. We can lower bound the variance after $M \in \mathbb{N}$ layers by:
    \begin{align}
    \operatorname{Var} \left[ \Tr\left(O\, \Phi^M\!(\rho_0)\right)\right] \geq \sum a_P^2 \frac{1}{C^\Lambda} D_{\max}^{2 {K^d \,\Delta}} q^{2 \Gamma + 2}.
\end{align}
     Here, $\Lambda=\max\left[   \delta^d \left( d (\delta-1)  + L\right),K^d(d (K-1) + L)\right]$ describes the volume of gates one fixes with $\delta = \frac{1}{2}\left(\frac{n}{n_r}\right)^{1/d} \!\!+ \,\frac{K}{2} $. $\Gamma =  \frac{ n_r }{n} \frac{K^d(d(K-1)+L)}{L}$ counts the number of qubits reset in this volume.  $\Delta =\max \left[ d (\delta-1)  + L,d (K-1)+L\right]
    $, denotes the depth needed to reach the reset qubit. The constant $C$ is; $C \leq \abs{\operatorname{CL}(2)}   \abs{\operatorname{CL}(1)}$ and $D_\text{max}$ is the largest diagonal entry of the noise channel in normal form.
\end{proposition}

\begin{figure}
    \centering
    \includegraphics{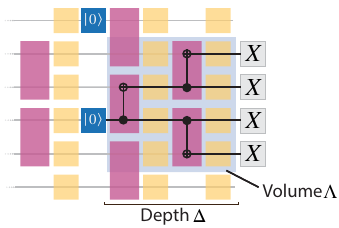}
    \caption{\textbf{Operator transformation example} A circuit illustrating how two-qubit Clifford gates can be fixed to transform observables. Specifically, the circuit maps the observable $X^{\otimes 4}$ to $I \otimes I \otimes X \otimes I$ where the $X$ acts on a reset qubit. In every circuit layer the support can be shrunk by at least one in one dimension. The required gate depth and the volume of fixed gates enabling this transformation are highlighted.}
    \label{fig:Variance_examplecircuit}
\end{figure}

\begin{proof} We lower bound the variance after applying $M$ jumps. 
First we consider the expectation value of the observable. Using Lemma~\ref{Lem:SingleGlobalOne} for the final layer of single-qubits gates:
\begin{align}
    \mathbb{E}\left[ \Tr \left( O\, \Phi^M\!(\rho_0)\right)\right] = a_{I^{\otimes n}} = \Tr(O) .\label{eq:Expectationvalue}
\end{align}
We now calculate the variance $
\operatorname{Var} \left[ \Tr ( O\, \Phi^M\!(\rho_0)) \right] =\mathbb{E} \left[ \Tr \left( O\, \Phi^M\!(\rho_0)\right)^2 \right] -\mathbb{E}\left[ \Tr \left( O\Phi^M(\rho_0)\right)\right]^2
$.
By applying the first point of Lemma~\ref{lem:SecondMomentSingle} to the last layer of gates we can rewrite:
\begin{align}
\mathbb{E}\big[\Tr(O\, \Phi^M\!(\rho_0))^2\big] &= \sum_{P \in \{I, X, Y, Z\}^{\otimes n}} a_P^2 \, \mathbb{E}\big[\Tr(P \Phi'^M(\rho_0))^2\big] \nonumber\\
&= a_{I^{\otimes n}}^2 + \sum_{P \in \{I, X, Y, Z\}^{\otimes n} \setminus \{I^{\otimes n}\}} a_P^2 \, \mathbb{E}\big[\Tr(P \Phi'^M(\rho_0))^2\big]. 
\end{align}
Using Eq.~\eqref{eq:Expectationvalue}, the variance evaluates to:
\begin{align}
      \operatorname{Var} \left[ \Tr ( O\, \Phi^M\!(\rho_0)) \right] = \sum_{P \in \{I, X, Y, Z\}^{\otimes n} \setminus \{I^{\otimes n}\}} a_P^2 \, \mathbb{E}\big[\Tr(P \Phi'^M(\rho_0))^2\big]. \label{eq:Varaince_writtenout_SI}
\end{align}
The squared expectation values of the Pauli strings can be treated separately. The unital noise channel $\mathcal{N}$ is sandwiched between Haar random gates allowing us to use their unitary invariance to bring the channel to normal form. The action of the adjoint noise channel $\mathcal{N}^*$ on a Pauli string is
\begin{align}
    \mathcal{N}^*(Q) = D_Q Q.
\end{align}
Since $\mathcal{N}$ is not fully depolarizing, the maximal factor satisfies $D_{\max} = \max_{Q \in \{X,Y,Z\}}(D_Q) > 0$.  We make use of the fact that the Clifford group forms a 2-design by Lemma~\ref{lem:Clifford2design}. While dealing with second moment quantities we can replace the integral over the Haar measure by a sum over the Clifford group. For any real function $g: U(2) \mapsto \mathbb{R}$, it is possible to lower bound the expectation value by the value for a single element of the group divided by the order of the group;
\begin{align}
    \underset{U \sim U(2)}{\mathbb{E}}\left[ g(U)^2\right] = \underset{U \sim \operatorname{CL}(2)}{\mathbb{E}}\left[ g(U) ^2\right] = \frac{1}{\abs{\operatorname{CL}(2)}}\sum_{U \in \operatorname{CL}(2)} g(U)^2 \geq \frac{1}{\abs{\operatorname{CL}(2)}} g(C_\text{fixed})^2. \label{eq:Lower_Clifford_fixed}
\end{align}
 We focus our attention on the Pauli operator $Q_\text{max}$ with maximal $D_Q$ that retains the highest magnitude under the action of the channel. There exists a operator in the Clifford group that transforms the nontrivial Pauli operators that make up $P$ to $Q_\text{max}$. Specifically, the Hadamard and S gates allow us to transform between the different Pauli operators,  $X = H^\dag Z H$ and $X = S^\dag Y S$. We can then lower bound
 \begin{align}
     \mathbb{E}\big[\Tr\left(P \Phi'^M(\rho_0)\right)^2\big] 
     &\geq \frac{1}{\abs{\operatorname{CL}(1)}^{ \vert P \vert}} \mathbb{E}\left[ \Tr (\mathcal{N}^{*\otimes n}(Q_\text{max}^{\otimes \operatorname{supp}(P)} )  \Phi''_{\left[LM,1 \right]}(\rho_0))^2\right] \\
         & = \left(\frac{D_{\operatorname{max}}^2}{\abs{\operatorname{CL}(1)}}  \right)^{  \vert P \vert} \mathbb{E}\left[\Tr\left( Q_\text{max}^{\otimes \operatorname{supp}(P)}\Phi''_{\left[LM,1 \right]}(\rho_0)\right)^2 \right].
 \end{align}
In each layer we pick up a factor larger than $D_\text{max}^{2\abs{P}}$ due to the noise channel $\mathcal{N}$ and $\text{CL}(1)^{-\abs{P}}$ due to fixing the single qubit gates. In total we accumulate a factor 
\begin{align}
    \left(\frac{D_{\operatorname{max}}^2}{\abs{\operatorname{CL}(1)} } \right)^{  \vert P \vert} \geq \left(\frac{D_{\operatorname{max}}^2}{\abs{\operatorname{CL}(1)} } \right)^{ K } .
\end{align}
We fix two local gates to backward evolve the $\abs{P}$-local observable to an observable on a single reset qubit, giving a total factor of $C \leq \abs{\operatorname{CL}(2)}   \abs{\operatorname{CL}(1)}$. We illustrate this process for the Pauli string $X^K$ in Fig.~\ref{fig:Variance_examplecircuit}. At each layer, we apply the transformation
$
\left[I^{\otimes K-1} \otimes \text{CNOT}^\dag \right] X^{\otimes K} \left[ I^{\otimes K-1} \otimes \text{CNOT} \right] = X^{\otimes K-1} \otimes I,
$ which effectively reduces the support of the observable by one qubit per layer.  
In the following we determine the necessary depth to transform the operator $P$ to a $Z$ operator on one of the qubits that is reset. 

For this transformation we need to fix the gates in the path. There are two cases to distinguish:
\begin{enumerate}
    \item No qubit in the path of the Pauli is reset. 
    \item One or multiple qubits are reset. 
\end{enumerate}
In the first case, no qubit in $\operatorname{supp}(P)$ is reset and we need to establish a connection to the closest reset qubit. We choose any direction $\Hat{e}_i$. 
In this direction in the worst case the operator is positioned symmetrically between two reset qubits. Consequently, the edge of the operator is at most $ \delta = \frac{1}{2}\left(\frac{n}{n_r}\right)^{1/d} \!\!+ \,\frac{K}{2}$ away. 
In every unitary step we can reduce the support of the operator by at least one in that dimension and we need to shrink the operator in $d$-dimensions. Therefore, the depth needed to transform the Pauli to a single-qubit Pauli will be at most $d (\delta-1)$. In addition, we need to fix the gates until the next reset occurs, which can be at most $L$ layers away. This means the total depth to contract the operator on a single-qubit is at most $ d (\delta-1) + L$. The volume of fixed gates is given by the product of the area and the depth. Since the area is upper bounded by $\delta^d$, this yields an upper bound on the volume of $\delta^d \left( d (\delta-1)  + L\right)$.

In the second case, a qubit in $\operatorname{supp}(P)$ is reset. We can transform the operator to a single-qubit Pauli on this qubit. This requires a depth smaller than $d(K-1) + L$ and has volume smaller than $K^d(d(K-1) + L)$. The maximum number of reset qubits in this volume are $\frac{K^d}{n/n_r} \frac{d(K-1)+L}{L}$.
We obtain the following upper bounds on the extent of gates we need to fix:
\begin{align}
    \Delta &\leq \max \left[ d (\delta-1)  + L,d (K-1)+L\right] , \\
    \Lambda &\leq \max\Biggr[   \delta^d \left( d (\delta-1)  + L\right),K^d(d (K-1) + L)\Biggr],\label{eq:contractionparams_SI}\\
    \Gamma &\leq \frac{ n_r }{n} \frac{K^d(d(K-1)+L)}{L}.
\end{align}

If a qubit in the path is reset we need to adjust the Clifford operations accordingly. Suppose that a qubit is reset at depth $i$. We choose the Clifford before to map $Q \rightarrow Z$.  The backward-evolved operator at this point then takes the form $Z \otimes R$, where $R$ is some Pauli string. The action of the reset then results in:
\begin{align}
       \mathbb{E}&\left[ \Tr\left(\mathcal{A}^* (Z \otimes R)\Phi_{\left[ L M - i,1 \right]}(\rho_0)\right) \right]
    =\mathbb{E}\left[ \Tr\left( ((q I+(1-q)Z) \otimes R)\Phi_{\left[ L M - i,1 \right]}(\rho_0)\right)^2 \right] \nonumber\\
    &=\mathbb{E}\left[ q^2\Tr\left( (I \otimes R)\Phi_{\left[ L M - i,1 \right]}(\rho_0)\right)^2\right]
    +(1-q)^2\mathbb{E}\left[\Tr\left( (Z \otimes R)\Phi_{\left[ L M - i ,1\right]}(\rho_0)\right)^2\right] \nonumber\\
    &+q(1-q)\mathbb{E}\left[\Tr\left( (I \otimes R)\Phi_{\left[ L M - i,1 \right]}(\rho_0)\right)\Tr\left( (Z \otimes R)\Phi_{\left[ L M - i ,1\right]}(\rho_0)\right) \right] \\
    &\geq q^2 \mathbb{E}\left[ \Tr\left( (I \otimes R)\Phi_{\left[ L M - i ,1\right]}(\rho_0)\right)^2\right]. \nonumber
\end{align}
In the last step we apply the first point of Lemma~\ref{lem:SecondMomentSingle} . The operator we further backwards transform further is $I \otimes R$. The reset of a qubit does not modify the depth required to transform the observable to a single-qubit observable.

In the final layer before the reset we choose Cliffords that transform the operator $Q$ to a $Z$ operator on a reset qubit. Let $LM -j$ denote the depth at which the contracted observable meets the reset qubit, satisfying $j \leq \Delta$, then
\begin{align}
    \mathbb{E}\big[\Tr(P \Phi'^M(\rho_0))^2\big] \geq \frac{1}{C^\Lambda} D_{\max}^{2 {K^d \Delta}} q^{2 \Gamma} \mathbb{E}\left[ \Tr\left( (Z \otimes I^{n-1})\mathcal{A}(\Phi_{\left[ L M - j,1 \right]}(\rho_0))\right)^2 \right]. \label{eq:Shrinking_SI}
\end{align}
We apply the adjoint map of the amplitude damping channel $\mathcal{A}$ to the Pauli string and use the first point of Lemma~\ref{lem:SecondMomentSingle}
\begin{align}
   \mathbb{E}\left[ \Tr\left(\mathcal{A}^* (Z \otimes I^{n-1})\Phi_{\left[ L M - j ,1\right]}(\rho_0)\right) \right]
    &=\mathbb{E}\left[ \Tr\left( ((q I+(1-q)Z) \otimes I^{n-1})\Phi_{\left[ L M - j,1 \right]}(\rho_0)\right)^2 \right] \nonumber\\
    &\geq q^2 \mathbb{E}\left[ \Tr\left( (I \otimes I^{n-1})\Phi_{\left[ L M - j,1 \right]}(\rho_0)\right)^2\right] = q^2. \label{eq:Damping_SI2}
\end{align}
In the final equality we used that $\Phi_{\left[ L M - j,1 \right]}(\rho_0)$ is a physical state and therefore has trace 1. We combine Eq.~\eqref{eq:Damping_SI2} and Eq.~\eqref{eq:Shrinking_SI} to lower bound the variance of an individual term. By summing all the individual terms in Eq.~\eqref{eq:Varaince_writtenout_SI}, we obtain the lower bound:
\begin{align}
    \operatorname{Var} \left[ \Tr\left(O\, \Phi^M\!(\rho_0)\right)\right] \geq \sum a_P^2 \frac{1}{C^\Lambda} D_{\max}^{2 {K^d \,\Delta}} q^{2 \Gamma + 2}.
\end{align}
\end{proof}
\section{Review of barren plateaus}
Here we provide a quick review of literature on unitary and noise-induced barren plateaus. For a more comprehensive overview see Ref. \cite{larocca_barren_2025}.
\begin{definition}[Barren plateau]
A cost function exhibits a barren plateau if the variance of the cost function gradients is exponentially concentrated with system size:
    \begin{align}
    \operatorname{Var}\left[ \partial_\mu C \right] = \mathcal{O}\left(\frac{1}{b^n} \right),
\end{align}
for some constant $b>1$. 
\end{definition}
Further, if the gradient is centered at zero then by Chebyshev's inequality measuring gradients larger than any constant becomes exponentially unlikely in system size. Detecting exponentially small changes and optimizing the loss would require an exponentially large number of measurements (shots), making the algorithm both inefficient and unscalable \cite{larocca_barren_2025}.

     In their seminal work, \citet{mcclean_barren_2018} showed that for quantum circuits that form a 2-design the gradient of the cost function vanishes exponentially with the number of qubits. 
     Geometrically local random quantum circuits satisfy this property for linear depth \cite{brandao_local_2016} in one dimension and depth $\mathcal{O}(n^{1/d})$ in $d$ dimensions \cite{harrow_approximate_2023}. Furthermore, also gradient free optimization fails as the entire cost landscape becomes flat \cite{arrasmith_effect_2021}. 
     Subsequent theory refined this picture by pinpointing when plateaus occur: for example, \citet{cerezo_cost_2021} proved that even a shallow layered ansatz will exhibit exponential gradient decay if the cost grows with system size (global observable), whereas using a local cost operator yields only a polynomially vanishing gradient provided the circuit depth grows only logarithmically with system size.
     Entanglement between visible and hidden units can be a further cause of barren plateaus \cite{ortiz_marrero_entanglement-induced_2021}.
     Similarly, Holmes et al. quantified that ansatz expressibility controls trainability: ans\"atze closer to 2-designs (more expressive) produce flatter loss landscapes and hence smaller gradient magnitudes \cite{holmes_connecting_2022}. These effects have been recently connected to the dimension of the dynamical Lie algebra the of the circuit ansatz \cite{ragone_lie_2024,fontana_characterizing_2024}.

In many cases noise further worsens the problem.
\citet{wang_noise-induced_2021} show that for local Pauli noise circuit expectation values and gradients deterministically concentrate. Specifically, the expectation value approaches that of the fully-mixed state exponentially in the depth of the circuit.
 In the work Ref. \cite{singkanipa_beyond_2025} the authors show that HS-contractive non-unital noise can lead to noise induced fixed points that are distinct from the fully mixed state and that some non-unital circuits can avoid noise-induced barren plateaus.
Nonunitary elements have emerged as a promising approach to address unitary barren plateaus, with numerical evidence presented in Ref. \cite{binkowski_barren_2023}.
 Recently, the authors \citet{mele_noise-induced_2024} proved that in the presence of single qubit non-unital noise parameterized quantum circuits avoid both noise induced and unitary barren plateaus. Specifically, they show that for local cost functions the last $\mathcal{O}(\log(n))$ layers remain trainable. This counterintuitive result highlights the critical dependence of NIBPs on noise type.
 Dissipation as a tool to mitigate barren plateaus has been established by \citet{sannia_engineered_2024}. They show that dissipation can turn a global cost function into a local cost function. 

The growth of entanglement of quantum circuits has been linked to the appearance of barren plateaus \cite{ortiz_marrero_entanglement-induced_2021,sack_avoiding_2022}. Haar random states exhibit volume-law entanglement \cite{page_average_1993}. Periodically measuring quantum circuits can limit the amount of entanglement present in the system. As the measurement rate is varied, such systems can transition between volume-law and area-law entangled phases, a phenomenon known as a measurement-induced phase transition \cite{skinner_measurementinduced_2019}. \citet{wiersema_measurementinduced_2023} numerically  consider random quantum circuits and observe a measurement-induced phase transition and further demonstrate that a sufficiently high measurement rate inhibits the suppression of gradients.

Shortly prior to our work, an independent study \cite{deshpande_dynamic_2024} has been released that analyzes variances of the cost function in the presence of dynamical operations, with findings that align with our conclusions. Our approach differs substantially from Ref. \cite{deshpande_dynamic_2024} both in methodology and in underlying assumptions. In particular, the authors show that the variance of the cost function can be lower bounded, and thereby that there exist parameters whose gradients do not concentrate \cite{arrasmith_equivalence_2022,miao_equivalence_2024}. In contrast, we identify explicit conditions on the circuit architecture and pinpoint which gradients avoid concentration. This distinction is crucial: rather than proving that some gradients remain trainable, we provide a constructive framework for designing circuits where gradients are guaranteed to be trainable.

\section{Absence of barren plateaus} \label{app:Barren_Lower_app}
Here we consider dissipative parameterized quantum circuits. Specifically, we assume some gates of the circuit are parameterized and these parameters are optimized by minimizing a cost function. We investigate the scaling of the gradients of this cost function with system size. The circuit depends on variational parameters $\boldsymbol{\theta} := (\theta_1, \dots, \theta_m) \in \mathbb{R}^m$. These parameters determine the operation of a subset of the two-qubit gates, which are of the form $\{ \exp(-i\theta_\mu H_\mu) \}_{\mu=1}^m$, where $H_\mu$ are two-local Hermitian operators satisfying $0<\|H_\mu\|_\infty \leq 1$. The left right invariance of the Haar random two-qubit gates allows us to absorb the action of the parameterized gate into the random gate. Therefore, the results of previous sections are still applicable. We consider a gate with parameter $\theta_\mu$ which is located at unitary layer $\mathcal{U}_i$. 
With these parameters we can introduce the \textit{cost function} associated to an observable $O \in \mathcal{B}(\mathcal{H)}$ as 
\begin{align}
    C(\bm{\theta}) \coloneqq \Tr \left( O\, \Phi^M\!(\rho_0) \right)
\end{align}


\begin{lemma}(Gradient commutator) \label{lem:Grad_commutator}
    Let $\mu \in [m]$. We consider the parameterized gate $\exp(-i \theta_\mu H_\mu)$ located at layer $LM-i$ of the circuit. 
    Then the gradient can be expressed as 
    \begin{align}\label{eq:GradientExpression}
        \partial_\mu C = i \, \Tr \big( \Phi_{[LM-i-1,1]}(\rho_0) \big[ H_\mu, \Phi^*_{[LM,i]}(O) \big] \big). 
    \end{align}
\end{lemma}
\begin{proof} We re-express the cost function
\begin{align}
    C = \Tr\left(  \Phi_{[LM,LM-i]}  \circ \Phi_{[LM-i-1,1]}(\rho_0) O\right) = \Tr\left(  \Phi_{[LM-i-1,1]}(\rho_0)\Phi_{[LM,LM-i]} ^*( O) \right).
\end{align}
The gradient with respect to $\theta_\mu$ evaluates to
\begin{align}
    \partial_\mu C &= \Tr\left(  \Phi_{[LM-i-1,1]}(\rho_0) \partial_\mu\Phi_{[LM,LM-i]} ^*( O) \right) \nonumber\\
    &= \Tr\left(  \Phi_{[LM-i-1,1]}(\rho_0) \partial_\mu \exp(i H_\mu \theta_\mu)\Tilde{\Phi}_{[LM,LM-i]} ^*( O) \exp(-i H_\mu \theta_\mu)\right) \nonumber\\
    &= i \Tr\left(  \Phi_{[LM-i-1,1]}(\rho_0) H_\mu\Phi_{[LM,LM-i]} ^*( O) \right)  - j \Tr\left(  \Phi_{[LM-i-1,1]}(\rho_0) \Phi_{[LM,LM-i]} ^*( O)H_\mu \right) \\
    &= i \Tr\left(  \Phi_{[LM-i-1,1]}(\rho_0) \left[ H_\mu,\Phi_{[LM,LM-i]} ^*( O)\right] \right). \nonumber
\end{align}
We use $\Tilde{\Phi}^*$ to denote the channel $\Phi^*$ with the parameterized unitary removed. In the second line we used that the adjoint channel has adjoint Kraus operators. 
\end{proof}
\begin{lemma}\label{lem:ExpGradZero}(Expectation value zero) 
The expectation value of the derivative of the cost function is zero. 
\begin{align}
    \mathbb{E}[\partial_\mu C] = 0.
\end{align}
\end{lemma}
\begin{proof}
We first express the expectation value of the gradient with Eq.~\eqref{eq:GradientExpression} and then use that single-qubit haar random gates form a global one design (Lemma~\ref{Lem:SingleGlobalOne}):
    \begin{align}
        \mathbb{E}[\partial_\mu C]  &= i \, \mathbb{E}\left[\Tr \left( \Phi_{[LM-i-1,1]}(\rho_0) \left[ H_\mu, \Phi^*_{[LM,i]}(O) \right] \right)\right] \nonumber\\
        & = \Tr \left( \Phi_{[LM-i-1,1]}(\rho_0) \left[ H_\mu, \mathbb{E}\left[\Phi^*_{[LM,i]}(O) \right]\right] \right)\\
        &= \Tr \left( \Phi_{[LM-i-1,1]}(\rho_0) \left[ H_\mu, \Tr\left(\Phi^*_{[LM,i]}(O) \right) /d I \right] \right) = 0. \nonumber
    \end{align}
    In the final equality, we use that any operator commutes with the identity.
\end{proof}

\begin{lemma}(Pauli mixing gradients (Ref.~\cite[Lemma 43]{mele_noise-induced_2024}) ) \label{lem:PauliMixingGradients}
Let $H_\mu$ be a 2-local Hamiltonian. Define the operator function 
\begin{align}
f_j(\cdot) := i \, \Tr \big( \Phi_{[1,k]}(\rho_0) \big[ H_\mu, \Phi^*_{[k,L-j]}(\cdot) \big] \big).
\end{align}
Then, the following holds:
\begin{enumerate}
    \item Let $O := \sum_{P \in \{I, X, Y, Z\}^{\otimes n}} a_P P$, where $a_P \in \mathbb{R}$ for any $P \in \{I, X, Y, Z\}^{\otimes n}$. We have
    \begin{align}
    \mathbb{E}[(f(O))^2] = \sum_{P \in \{I, X, Y, Z\}^{\otimes n}} a_P^2 \, \mathbb{E}[(f(P))^2]. \label{eq:C12}
    \end{align}
    
    \item Moreover, for any $P \in \{I, X, Y, Z\}^{\otimes n}$, we have
    \begin{align}
    \mathbb{E}[(f(P))^2] = \frac{1}{3^{|P|}} \sum_{\substack{Q \in \{I, X, Y, Z\}^{\otimes n} \vert\\ \operatorname{supp}(Q) = \operatorname{supp}(P)}} \mathbb{E}[(f(Q))^2].
    \end{align}
\end{enumerate}

\end{lemma}
For proof see Ref.~\cite{mele_noise-induced_2024}.

\begin{lemma}(Gradient zero outside the light cone)\label{lem:LightconeGrads}
    If $H_\mu$ is outside the inverse light cone of the observable $O$ the gradient vanishes $\partial_\mu C = 0$.
\end{lemma}
\begin{proof}
We use Lemma~\ref{lem:Grad_commutator} to express the gradient 
\begin{align}
        \partial_\mu C = i \, \Tr \big( \Phi_{[LM-i-1,1]}(\rho_0) \big[ H_\mu, \Phi^*_{[LM,i]}(O) \big] \big). 
\end{align}
If $H_\mu$ is outside the inverse light cone of $O$, $\Phi^*_{[LM,i]}(O)$ is equal to the identity on the support of $H_\mu$. As the identity commutes with any operator the commutator is equal to zero. 
\end{proof}
\begin{theorem}(Lower bound on the variance of the partial derivative)\label{thm:lower_derivative}
Consider a circuit as described in Sec.~\ref{sec:DisCircuitModel} where every $L$ layers $n_{r}$ qubits are reset with probability $q$. The circuit is defined on a lattice of $n$ qubits and dimension $d$. Let $\mu \in [m]$. We consider the parameterized gate $\exp(-i \theta_\mu H_\mu)$ located at layer $LM-i$. Define the cost function $C(\bm{\theta}) = \Tr \left( O\, \Phi^M\! (\rho_0) \right)$ for $\rho_0$ an arbitrary initial state and $O=\sum_{P \in \{I, X, Y, Z\}^{\otimes n}} a_P P $ an observable with diameter $K$. Then if the support of $O$ is in the light cone of the gate
\begin{align}
    \operatorname{Var} \left[ \partial_\mu C\right] \geq \sum_{P \in \{I, X, Y, Z\}^{\otimes n}} a_P^2 \frac{D_{\max}^{2K^d (i -1) +2 \Delta}}{C^{K^d (i -1)+\Lambda}}  \frac{\norm{H_{\mu}}_\infty^2}{64} q^{2\Gamma+2} (1-q)^{K^d \frac{i-1}{L} \frac{n_r}{n}}
\end{align}
Here, $\Lambda=\max\left[   \delta^d \left( d (\delta-1)  + L\right),K^d(d (K-1) + L)\right]$ with $\delta = \frac{1}{2}\left(\frac{n}{n_r}\right)^{1/d} \!\!+ \,\frac{K}{2} $. $\Gamma =  \frac{ n_r }{n} \frac{K^d(d(K-1)+L)}{L}$ and $\Delta =\max \left[ d (\delta-1)  + L,d (K-1)+L\right]
    $, and $C \leq \abs{\operatorname{CL}(2)} \abs{\operatorname{CL}(1)}$. $D_{\max} = \max_{Q \in P(1)} (D_Q)$ is the direction that is least shrunk by the noise. 
\end{theorem}
The main idea of the proof is to fix gates to backward transform the observable to an observable with support on the qubits acted on by the gate.
 Finally, the variance of the gradient is related to the variance of an observable which can be lower bounded using Proposition~\ref{prop:variance_lowerbound}.
\begin{proof} 
Using Lemma~\ref{lem:ExpGradZero} we can express the variance of the gradient as
\begin{align}
    \operatorname{Var} \left[ \partial_\mu C\right] = \mathbb{E}\left[\left( \partial_\mu C\right)^2\right] = \mathbb{E}\left[\left( f_0(O) \right)^2\right],
\end{align}
where $f_j$ is defined as in Lemma~\ref{lem:PauliMixingGradients}. We use the first point of Lemma \ref{lem:PauliMixingGradients} to re-express:
\begin{align}
    \mathbb{E}\left[\left( f_0(O) \right)^2\right] = \sum_{P \in \{I, X, Y, Z\}^{\otimes n}} a_P^2 \mathbb{E}\left[\left( f_0(P) \right)^2\right].
\end{align}
In the following step, we track how one of these Pauli observables are affected by the noise model. Now we consider the action of the adjoint of the noise channel $\mathcal{N}$, it transforms
\begin{align}
    \mathcal{N}^{*} (P) = \bigotimes_{j \in \operatorname{supp}(P)} (c_{Q_j} + D_{Q_j} Q_j) =  \prod_{j \in \operatorname{supp}(P)} D_{Q_j} P_{Q_j}.
\end{align}
Here we used that $\mathcal{N}$ is unital and therefore $c_Q = 0$ for all $Q$. Since noise shrinks different Pauli components differently, we choose the least shrunk direction $D_{\max}$ to obtain the most favorable bound i.e. $D_{\max} = \max_{Q \in P(1)} (D_Q)$ associated to Pauli operator $Q_{\max}$. By our assumption, the noise channel is not the fully depolarizing channel and therefore $D_\text{max}$ nonzero. 
To proceed we define 
\begin{align}
f'_j(\cdot) := i \, \Tr \big( \Phi_{[1,k]}(\rho_0) \big[ H_\mu, \Phi'^*_{[k,L-j]}(\cdot) \big] \big).
\end{align}
here $\Phi'_{[k,L-j]}$ denotes the map $\Phi_{[k,L-j]}$ with single-qubit unitaries removed Similarly, we define $f''$ using the map $\Phi''^*_{[k,L-j]} $, which excludes both noise and single-qubit unitaries. We now choose Clifford operations to map every entry of $P$ to $Q_\text{max}$, then
\begin{align}
    \mathbb{E}\left[\left( f_0(P) \right)^2\right] \geq\left(\frac{ D_\text{max}^2 }{\operatorname{CL}(1)}\right)^{\abs{P}} \mathbb{E}\left[ \left( f_0''(Q_{\text{max}}^{\otimes\operatorname{supp}(P)}) \right)^2\right].
\end{align}
Further, $\abs{P} \leq K^d$. We proceed fixing Clifford gates to transform the operator $Q$ into an operator with support on $H_\mu$. We refer to this backward-evolved operator $Q_b$. In this path we make sure that the support of the operator does not grow meaning it acts on less than $K^d$ qubits. Since $H_\mu$ is in the backward light cone of the observable $O$ we are guaranteed that such a path exists.
This path is $i -1$ layers deep that require us to fix less than $K^d (i -1)$ gates. 

When a qubit along the path is reset, we first apply a Clifford transformation that maps the relevant Pauli operator to $Z$. After reset, the operator becomes
$\mathcal{A}(Z)^{\otimes 2} \otimes R^{\otimes 2} = (q I + (1 - q) Z)^{\otimes 2} \otimes R^{\otimes 2},$ where $R$ is a Pauli string acting on the other qubits. Expanding the reset part yields four terms. However, by the Pauli-mixing lemma~\ref{lem:Pauli2Mixing}, only symmetric positive terms survive. This allows us to discard the identity term, keeping only the $Z$ which carries a prefactor of $(1-q)^2$. Keeping the non-identity term ensures that the backwards-evolved Pauli string $Q_b$ is not the identity allowing for nontrivial commutators. In total, there are less than $\gamma = \frac{i-1}{L} K^d \frac{n_r}{n}$ resets in the path.
Then 
\begin{align}
\label{eq:backwardstogate}
    \mathbb{E}\left[(f_0(P))^2\right] \geq \left(\frac{D_{\max}^2 }{C} \right) ^{K^d (i -1)} (1-q) ^{2 \gamma} \mathbb{E}\left[ \Tr \left( \Phi_{[LM - (i-1) ,1]}(\rho_0) i \left[ H_\mu , Q_{b}\right]\right)^2 \right]
\end{align}

We can expand the rotation generator in the two-qubit Pauli basis $H_\mu = \sum_{R \in \{I, X, Y, Z\} ^{\otimes 2}} b_{R} R$. Substituting
\begin{align}
    \mathbb{E}\left[ \Tr \left( \Phi_{[LM - (i-1) ,1]}(\rho_0) i \left[ H_\mu , Q_{b}\right]\right)^2 \right]  = \sum_{R \in \{I, X, Y, Z\}^{\otimes 2}} b_R^2 \mathbb{E}\left[ \Tr \left( \Phi_{[LM - (i-1) ,1]}(\rho_0) Q_R\right)^2 \right]
\end{align}
where we defined $Q_R \coloneqq i \left[ R , Q_b\right]$. In the transformation, we made sure that the support of the operator does not grow and therefore $\operatorname{diam(Q_R) \leq K}$. We know that since $H_\mu$ is traceless and non-zero there exists a $b_{\tilde{R}} = \max_{R \in P(2)} [{b_R}]>0$. We choose the transformation such that the backward $Q_b$ does not commute with $\tilde{R}$:
\begin{align}
    \mathbb{E}\left[ \Tr \left( \Phi_{[LM - (i-1) ,1]}(\rho_0) i \left[ H_\mu , Q_b\right]\right)^2 \right] \geq b_{\tilde{R}}^2 \mathbb{E}\left[ \Tr \left( \Phi_{[LM - (i-1) ,1]}(\rho_0) Q_{\tilde{R}}\right)^2 \right].
\end{align}
Since $b_{\tilde{R}}^2$ is the maximal coefficient in the Pauli basis we can lower bound it by the average coefficient
\begin{align}\label{eq:coeff}
    b_{\tilde{R}}^2 \geq \frac{1}{16} \sum b_R^2 = \frac{1}{64} \norm{H_{\mu}}_2^2 \geq \frac{1}{64} \norm{H_\mu}_\infty^2.
\end{align}
Now we apply Proposition~\ref{prop:variance_lowerbound} to lower bound the variance of the observable $Q_{\tilde{R}}$. The observable has a diameter smaller than $K$
\begin{align} \label{eq:lowesinglegrad}
    \mathbb{E}\left[ \Tr \left( \Phi_{[LM - (i-1) ,1]}(\rho_0) Q_{\tilde{R}}\right)^2 \right] \geq \frac{1}{C^\Lambda} D_{\max}^{2 {K^d \,\Delta}} q^{2 \Gamma + 2}.
\end{align} 
Combining the results in Eq.~\eqref{eq:backwardstogate}, Eq.~\eqref{eq:coeff} and Eq.~\eqref{eq:lowesinglegrad} we find that 
\begin{align}
    \operatorname{Var}\left[ \partial_\mu C \right] 
    &\geq \sum_{P \in \{I, X, Y, Z\}^{\otimes n}} a_P^2  \mathbb{E}\left[(f_0(P))^2\right] \nonumber\\
    &\geq  \sum_{P \in \{I, X, Y, Z\}^{\otimes n}} a_P^2  \left(\frac{D_{\max}^2 }{C} \right) ^{K^d (i -1)} (1-q)^{2 \gamma}\mathbb{E}\left[ \Tr \left( \Phi_{[LM - (i-1) ,1]}(\rho_0) i \left[ H_\mu , Q_b\right]\right)^2 \right] \nonumber\\
    &\geq \sum_{P \in \{I, X, Y, Z\}^{\otimes n}} a_P^2 \left(\frac{D_{\max} ^2}{C} \right) ^{K^d (i -1)} (1-q)^{2 \gamma} b_{\tilde{R}}^2\frac{1}{C^{\Lambda}} D_\text{max}^{2\Delta}q^{2\Gamma+2} \\
    &\geq \sum_{P \in \{I, X, Y, Z\}^{\otimes n}} a_P^2 \frac{D_{\max} ^{2 K^d(i-1) + 2\Delta} }{C^{K^d(i-1) + \Lambda}}  \frac{\norm{H_{\mu}}_\infty^2}{64} q^{2\Gamma+2} (1-q)^{2 \gamma}.\nonumber
\end{align}

\end{proof}

\begin{corollary}(Absence of barren plateaus logarithmic depth jumps) \label{cor:AbsenceBarrenLog}
    Consider a circuit as described in Section~\ref{sec:DisCircuitModel}, with jumps of logarithmic depth $L = \mathcal{O}(\log(n))$ on a $d$-dimensional lattice. Let $\exp(-i \theta_\mu H_\mu)$ be a parameterized gate located at layer $LM - i$, where $i = \mathcal{O}(\log(n))$. If the number of reset qubits scales with system size, that is, $n_r = \Omega(n)$, and the support of the gate lies within the inverse light cone of a local cost function with bounded support $K = \mathcal{O}(1)$, then the gradient of the cost function with respect to $\theta_\mu$ does not exhibit exponential concentration. We can lower bound the variance of the gradient by 
    \begin{align}
        \operatorname{Var}\left[\partial_\mu C \right] \geq \Omega\left(\frac{1}{\operatorname{poly(n)}}\right).
    \end{align}
\end{corollary}
In particular, the variance of the gradient is independent of the depth of the circuit $L \cdot M$.
Provided the jumps are of constant depth, $L = \mathcal{O}(1)$ the support of the observable $O$ is bounded and the gate of which we calculate the gradient is a constant number of layers of from the measurement get a constant lower bound. 

\begin{figure}
    \centering
    \includegraphics[width=0.9\linewidth]{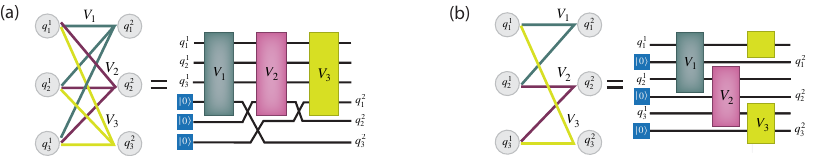}
    \caption{\textbf{Dissipative perceptron-based quantum neural networks.} (a) Global perceptrons. Dissipative quantum neural networks that rely on deep global perceptrons, analogous to fully connected classical networks, suffer from barren plateaus. (b) Geometrically local perceptrons. By restricting the locality of the operations to be geometrically local dissipative quantum neural networks can escape the barren plateau problem. }
    \label{fig:DQNN}
\end{figure}

Dissipative or perceptron-based quantum neural networks associate qubits with the nodes of the network. In each layer, the qubits from the previous layer are discarded, and the connections between layers--referred to as perceptrons---define the structure of the model \cite{altaisky_quantum_2001,sagheer_autonomous_2013,schuld_quest_2014,siomau_Quantum_2014,torrontegui_unitary_2019,tacchino_artificial_2019,beer_training_2020,wilkinson_evaluating_2022}.

\citet{sharma_trainability_2022} show that such networks can suffer from barren plateaus, particularly when the perceptrons are global (i.e., their support grows with system size) and deep, such that they approximate a 2-design. An example of a single layer from such a circuit is shown in Fig.~\ref{fig:DQNN}(a). Moreover, even shallow and local perceptrons can exhibit barren plateaus if the cost function is global. The corollary~\ref{cor:AbsenceBarrenLog} can be extended to show that geometrically local perceptrons with local cost functions do not suffer from barren plateaus.  We illustrate a single layer of such a configuration in one spatial dimension in Fig.~\ref{fig:DQNN}(b). However, the parameterized gate must be applied at most at logarithmic depth from the measurement as gradients are exponentially suppressed with this depth.

\section{Numerical results}
We numerically investigate the presence of barren plateaus. To isolate the effect of unitary barren plateaus we consider a setting without noise. In this setting we go beyond the analytic results by correlating parameters between layers and by going beyond the brickwork structure. We then consider noise-induced barren plateaus and investigate the effect of noise in toric code ground state preparation.

\subsection{Unitary barren plateaus}\label{app:Barren_nonoise}
We simulate random parameterized quantum circuits to numerically investigate the presence of unitary barren plateaus. The circuit ansatz we consider is built from two-qubit bricks of the form:
\begin{align}
    U(\bm{\theta}) = R_Y(\theta_1) \otimes R_Y(\theta_2)  \operatorname{CNOT}R_X(\theta_3) \otimes R_X(\theta_4),
\end{align}
denoting rotations generated by $A$ as $R_A(\theta) = e^{-\frac{i}{2} \theta A } $. 
We consider the architecture in one dimension with open boundary conditions meaning the bricks alternatingly connect qubits $(1,2)(3,4)\dots(n-1,n)$ and qubits $(2,3)(4,5)\dots(n-2,n-1)$. The ends of the one dimensional chain are not connected.

We go beyond the brickwork structure with a QAOA inspired circuit as is used in many quantum machine learning, state preparation, and optimization tasks \cite{verdon_quantum_2019,wecker_progress_2015,farhi_quantum_2014}. The circuit layers consist of $R_{ZZ}$ $R_X$ and $R_y$ gates, while $R_{ZZ}$ is applied with even and odd first qubits.

The dissipative ansatz differs from the unitary one by incorporating non-unital amplitude damping channels. This channel is applied to every second qubit after 5 layers. We assume that the reset operation is perfect with output state $\ketbra{0}{0}$.

The observable we consider is $O = Z_1$ on the second qubit with associated cost is $C(\bm{\theta}) = \Tr (\rho(\bm\theta) O)$. We calculate the gradient with respect to a parameter $\theta_\mu$ that appears in the last layer. We clearly distinguish unitary from noise-induced barren plateaus in the simulation by removing all noise except the amplitude damping channel used to reset ancilla qubits. As a consequence, the suppression of the gradient is due to the space the ansatz explores rather than the noise. 

In Fig.~\ref{fig:Brick_reset5} we plot the variance of the gradient for the brickwork circuit. The layers are repeated 40 times and in the dissipative ansatz every 5 layers every second qubit is reset. In panel (a) we show the variance of the gradient for all unitaries parameterized individually. Panel (b) goes beyond the assumptions of our proof, here the parameters of the unitaries are repeated every five layers reminiscent of the circuit shown in Fig.~\ref{fig:Maindoublewide} (c). Both panels qualitatively show the same behavior; the variance of the gradient is exponentially suppressed for the unitary ansatz but not for the dissipative ansatz. Consequently, the unitary ansatz suffers from barren plateaus while the dissipative ansatz is barren plateau free. The variance of the QAOA ansatz is shown in~\ref{fig:QAOA_reset5}. The behavior closely matches the brickwork circuit. The gradient of the unitary ansatz is exponentially suppressed in system size while that of the dissipative circuit remains constant.
\begin{figure*}
    \centering
    \includegraphics{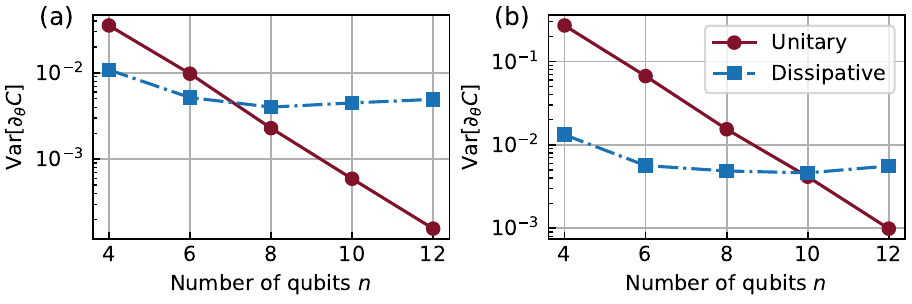}
    \caption{
    \textbf{Variance of the brickwork circuit} We plot the gradient variance for a one-dimensional brickwork ansatz with 40 layers. The dissipative ansatz resets every second qubit after 5 layers. In (a), all layer parameters are random, while in (b), they repeat every five layers. The gradient is measured for a local cost function based on the observable $Z$ on the second qubit. While the unitary gradient decays exponentially with system size, the gradient of the dissipative ansatz stays constant.}
    \label{fig:Brick_reset5}
\end{figure*}

\begin{figure*}
    \centering
    \includegraphics{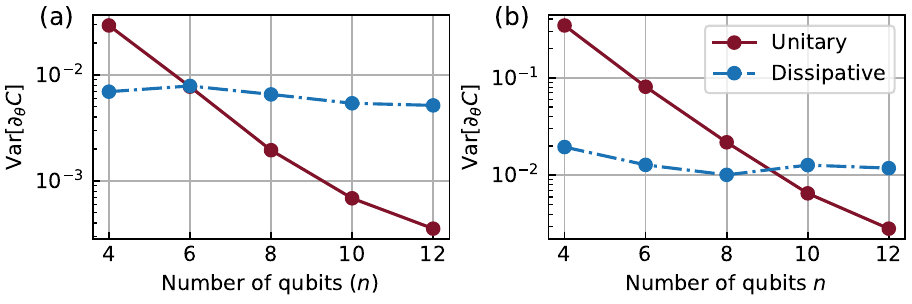}
    \caption{\textbf{Variance of the QAOA ansatz} We plot the gradient variance for a one-dimensional QAOA ansatz with 40 layers. The dissipative ansatz resets every second qubit after 5 layers. In (a), all layer parameters are random, while in (b), they repeat every five layers. The gradient is measured for a local cost function based on the observable $Z$ on the second qubit. While the unitary gradient decays exponentially with system size, the gradient of the dissipative ansatz stays constant.}
    \label{fig:QAOA_reset5}
\end{figure*}

\subsection{Noise induced barren plateaus}\label{app:toric_NIBP}
Here we consider the effect of noise numerically and present a problem that can only be solved through the use of dissipation. We consider ground state preparation of the toric code Hamiltonian. The toric code is a type of quantum error-correcting code defined on a two-dimensional lattice, arranged on the surface of a torus. It is a topological code, where qubits are placed on the edges of the lattice, and the code uses a set of quasi-local stabilizers to detect errors. These stabilizers are associated with the plaquettes (faces) and vertices (crosses) of the lattice. The logical states are the ground states of the toric code Hamiltonian
\begin{align}
    H = -\sum_\nu A_{\nu} -\sum_\beta B_{\beta}.
\end{align}
Here, $A_\nu = X_{\nu,1}X_{\nu,2}X_{\nu,3}X_{\nu,4}$ describes the stabilizer associated with plaquette $\nu$ and $B_\beta = Z_{\beta,1}Z_{\beta,2}Z_{\beta,3}Z_{\beta,4}$ the stabilizer associated with vertex $\beta$. Each stabilizer is four local connecting four neighboring qubits. 

The goal of both the dissipative and unitary algorithms is to prepare the ground state of the toric code. The toric code encodes 2 logical qubits; therefore, its ground-state manifold contains four degenerate states. 
We assume qubits arranged on a torus with local connectivity. In the dissipative setting we additionally assume that for each plaquette or vertex we couple to one ancilla. These ancillae can be reset during the operation of the algorithm. 

A quantum circuit that, starting from a product state, prepares the ground state of the toric code with quasi-local gates on the lattice must have a depth of at least $\Omega( \sqrt{n})$ where $n$ is the number of qubits in the code due to the Lieb-Robinson bound \cite{dennis_topological_2002,bravyi_lieb-robinson_2006}. We assume that we do not know the form of the unitary circuit that prepares the ground state. Instead, we want to train a variational quantum eigensolver (VQE) with a QAOA ansatz to prepare the ground state \cite{peruzzo_variational_2014}. We interleave layers of the form 
\begin{align}
U(\bm{\theta}) &= \prod_{j=L} ^1 U_l(\bm{\theta}_l),\\
    U_l(\bm{\theta}_l)&=\prod_\nu e^{-\frac{i}{2}\theta_{l,A} A_\nu}\prod_\beta e^{-\frac{i}{2}\theta_{l,B} B_\beta}.
\end{align}
The depth of this VQE will have to scale with $L = \sqrt{n}$. Physically implementing gates introduces errors. We model these errors by introducing depolarizing noise for every rotation $e^{-\frac{i}{2}\theta_{l,A} A_\nu}$ on all qubits involved in the operation. 
\begin{figure*}
    \centering
    \includegraphics[width=\textwidth]{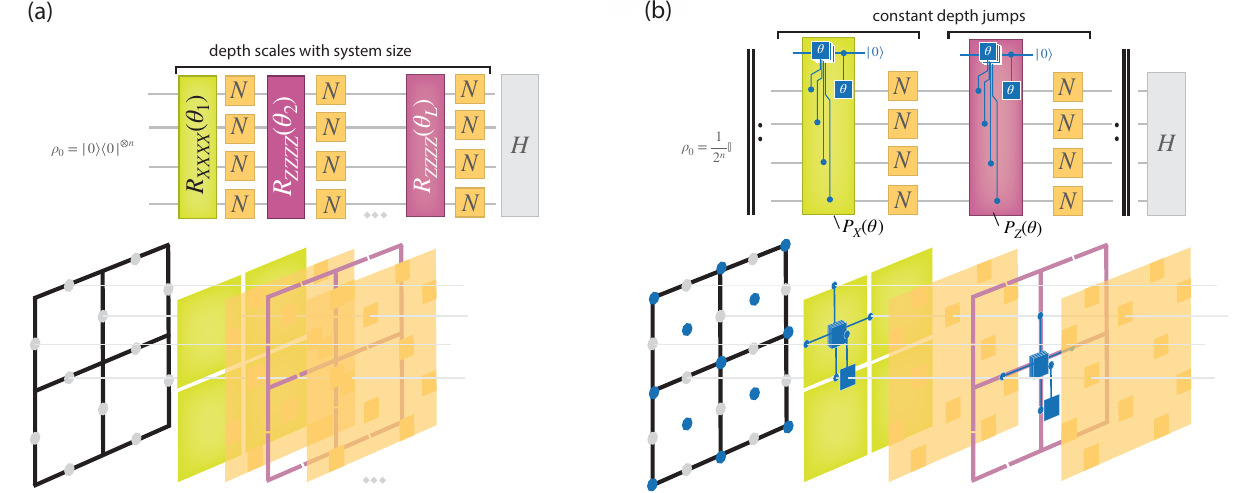}
    \caption{\textbf{Circuits for toric code ground state preparation} a. The unitary circuit is initialized in the product state $\ket{0^{\otimes n}}$. The ansatz consists of parameterized $XXXX$ and $ZZZZ$ rotations generated by the Hamiltonian terms. Following every rotation we apply depolarizing noise $N$. The rotations are arranged according to the connectivity of the toric code Hamiltonian. b. The dissipative circuit is initialized in the fully mixed state. The parameterized $P_X$ and $P_Z$ operations follow the connectivity of the toric code. They consist of four $CR_X$ rotations controlled by the qubits of each cross (plaquette) with an ancilla as the target followed by a $CR_X$ operation controlled by the ancilla targeting one qubit of the cross (plaquette). For every operation we apply depolarizing noise.}
    \label{fig:ToricCircuit}
\end{figure*}

The ground state of the toric code can be prepared by dissipative operators of constant weight \cite{dennis_topological_2002}.
The total time required to produce the ground state is not reduced in the dissipative setting \cite{konig_Generating_2014}. However, the depth of an individual jump is limited to a constant depth and these local jumps can be realized by a quantum circuit of constant depth. Therefore, we expect the ansatz to not be affected from noise induced and unitary barren plateaus. 

The dissipative ansatz we consider is inspired by circuits that map to ground states of the individual stabilizer terms. Usually, such a circuit is constructed from the following primitives: To project transform to the +1 eigenstate of a $Z$ stabilizer $B_\beta = Z_{\beta,1}Z_{\beta,2}Z_{\beta,3}Z_{\beta,4}$ associated with vertex $\beta$, we first map the parity to an ancilla qubit using four CNOT operations with the involved qubits as control and the ancilla as a target. Then by applying a CNOT operation controlled by the ancilla the parity of is switched conditioned on the system parity being odd. This operation maps any state into an even parity state making it a +1 eigenstate of the vertex operator. To prepare ground states of the plaquette terms the stabilizer readout needs to be sandwiched between Hadamard operations.

We parameterize this ansatz by replacing all CNOT gates by $\text{C}R_X(\theta)$ gates with a $S$ operation on the target qubit. For $\theta = \pi/2$ they perform a CNOT operation. We denote the combined operation $P_X(\theta)$ and $P_Z(\theta)$. Following the application of these operators depolarizing noise acts on all involved qubits.

The depth of the layers that implement the jumps is independent of the number of physical qubits. The dissipative ansatz in conjunction with the constant depth of the jumps guarantees that the ansatz does not suffer from barren plateaus. In particular, NIBPs also do not present an issue.
\begin{figure*}
    \centering
    \includegraphics{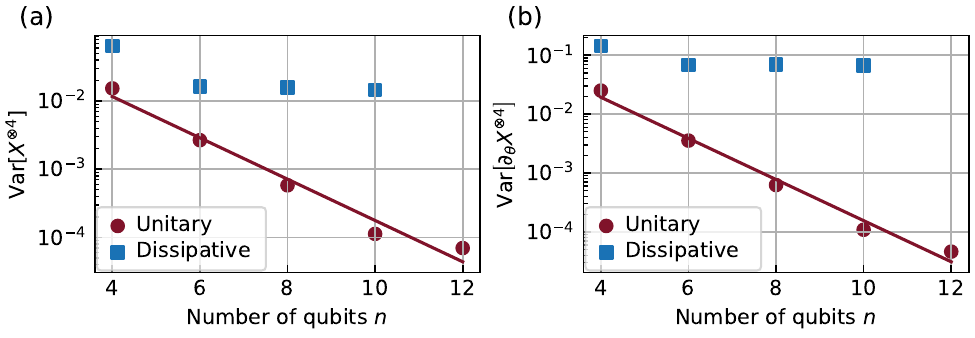}
    \caption{\textbf{Noise-induced variance concentration} The unitary circuit is initialized in the state $\ket{+}^{\otimes n}$ and then subjected to $n/2$ layers of unitaries, interleaved with depolarizing noise at a rate of 0.1. In contrast, the dissipative circuit starts in the fully mixed state and undergoes repeated application of two jumps of constant depth. Due to the high noise rate, the steady state of the dissipative process can be approximated by a circuit of depth 20. (a) The variance of a local observable of the unitary circuit is exponentially suppressed in system size. (b) Also the variance of the gradient is exponentially suppressed for the unitary case. We conclude that the unitary circuit suffers from NIBPs while the dissipative circuit does not.}
    \label{fig:grad_concentration}
\end{figure*}

\begin{figure}
    \centering
    \includegraphics{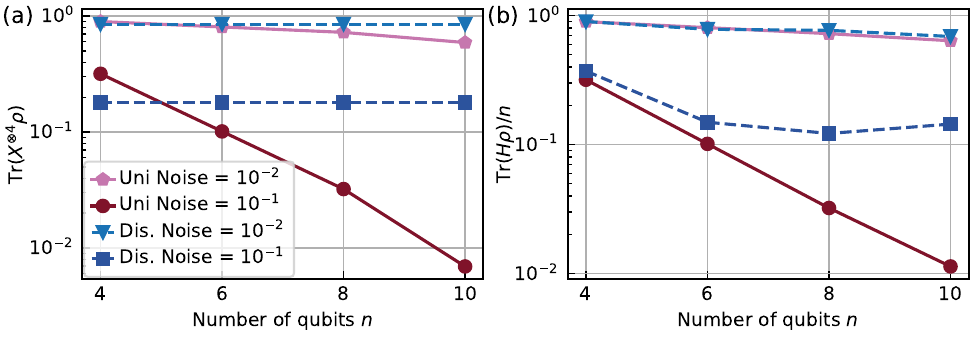}
    \caption{\textbf{Noise-induced deterministic concentration} In the presence of noise the expectation values of any unitary circuits exponentially approach that of the fully mixed state in the depth of the circuit. Dissipative circuits can avoid this concentration. In Panel (a) we compare the expectation value of one term of the toric code Hamiltonian while in (b) we plot the expectation value of the Hamiltonian. We compare the trained unitary ansatz to the trained dissipative ansatz. The expectation value of the unitary circuit exponentially approaches the noise-induced fixed point in the number of qubits while that of the dissipative circuit does not decay exponentially.}
    \label{fig:trainedconcentration}
\end{figure}
In Fig.~\ref{fig:grad_concentration} we show numerical evidence for the absence of noise-induced barren plateaus in dissipative learners. In panel (a) we plot the expectation value of one Pauli string for the unitary and dissipative circuit and different system sizes . We consider a rectangular lattice of height $2$ and width $n/2$. This allows us to simulate more system sizes (compared to square lattices) before running into computational bottlenecks. In this setting, the depth of the unitary circuit has to increase as $n/2$ while the dissipative jumps remain constant depth. As a consequence, the expectation value of the unitary exponentially concentrates at the value of the fully mixed state. The same behavior can be seen for the circuit gradients in panel (b). The variance of the gradients of the unitary circuit exponentially decay with system size, while those of the dissipative circuit do not decay. 

Noise not only leads to probabilistic concentration but also to deterministic concentration we show this in Fig.~\ref{fig:trainedconcentration}. We train both the dissipative and unitary ansatz for different system sizes $n$ of the toric code on rectangular lattice of height $2$ and width $n/2$. Then we evaluate the expectation value of a single Pauli string (Panel (a)) and of the Hamiltonian (Panel (b)). The expectation value of the single string remains constant with system size for the dissipative ansatz but is suppressed in for the unitary one. The Hamiltonian expectation value normalized by the number of terms in the Hamiltonian ideally remains at one. For the dissipative ansatz the expectation value shows an initial drop followed by a plateau, while that of the unitary ansatz exponentially drops with system size. These results show that toric code ground states can be prepared dissipatively but not unitarily.


\section{Expression for the steady-state and numerical simulation}
Dissipative quantum algorithms can have interesting noise-induced fixed points that are distinct from the fully mixed state \cite{mi_stable_2024,singkanipa_beyond_2025}. We model these algorithms as unitary transformations interleaved with non-unital noise processes. Here we derive closed form expressions for their fixed points. This expression can be used for analytical analysis of the steady state. Furthermore, by evaluating it numerically we can observe how the simulation of the evolution converges to the steady state.

We use the coherence vector formalism introduced in SI.~\ref{app:coherence_vec_pic}, and consider a circuit that consists of $M$ repetitions of a circuit primitive, which we will refer to as the \textit{jump}. One such jump consists of initial reset of one or multiple ancillary qubits, $L$ unitaries $U_j(\theta_j)$ interleaved with nonunitary channels $\mathcal{N}_j$ (Fig.~\ref{fig:Maindoublewide} (c)). We assume that these nonunitary channels are contractive, which ensures convergence to a steady state. The complete circuit consists of a set of $L$ unitaries $\{U_i \}_{j=1\dots L}$ and $L$ channels $\{\mathcal{N}_j\}_{j= 1,\dots,L}$. In contrast to SI.~\ref{sec:DisCircuitModel} here the unitaries are correlated between the jumps. 

The density matrix after the $j$-th application of the jump is denoted
\begin{align}
    \rho_j &= \left[ \mathcal{N}_j \circ \mathcal{U}_j (\theta_j) \right]\rho_{j-1}, \quad \forall j.
\end{align}
The coherence vector corresponding to $\rho_j$ is denoted as $\bm{v}_j$. In the coherence vector picture, each unitary is associated with an orthogonal transformation $O_j(\theta)$ and each channel to an affine transformation $v' = M_j \bm{v}  + \bm{c}_j$. We denote the combined operation of the unitary and the channel $v' =M_j O_j(\theta) \bm{v} + \bm{c}_j=\Omega_j \bm{v} + \bm{c}_j$
The coherence vector of state $\rho_j$ is
\begin{align}
    \bm{v}_j &= \Omega_j \dots \Omega_1 \bm{v}_0 + \bm{d}_j, \label{eq:ssequation}\\
    \bm{d}_j &= \sum^{j-1}_{r=1} \prod^{r+1}_{s=j}\Omega_s \bm{c}_r + \bm{c}_j .
\end{align}
After $M$ applications of the circuit primitive the state can be expressed as
\begin{align}
    \bm{v}^M = \left(\prod_{j=L} ^1 \Omega_j  \right)^M \bm{v}_0 + \sum_{s=0}^{M-1} \left(\prod_{j=L} ^1 \Omega_j  \right)^s \bm{d}_L.
\end{align}
In the limit $M \gg 1$ we can disregarded the first term as by Lemma~\ref{lem:TP_channel_evsmallerone} the eigenvalues of $\Omega_j$ are smaller or equal to one. Using the same assumption we can apply the geometric series for matrices and find an expression for the state after $M$ applications,
\begin{align}
    \bm{v}^M &= \left( I- \prod_{j=L} ^1 \Omega_j  \right)^{-1} \left( I - \left(\prod_{j=L} ^1 \Omega_j\right)  ^{M}\right) \bm{d}_L,
\end{align}
and for the steady state,
\begin{align}
    \bm{v}^\infty &=  \left( I-\prod_{j=L} ^1 \Omega_j  \right)^{-1} \bm{d}_L. \label{eq:SteadyState}
\end{align}
This expression allows for a more thorough analytical analysis of the steady-state.

Numerically evaluating the steady state can require deep circuits if noise strengths are low. In addition, one has to repeat the simulations multiple times to ensure that one has indeed prepared the steady-state. 
Eq.~\eqref{eq:SteadyState} gives a closed-form solution for the state after an infinite number of jumps. When we numerically evaluate the steady-state of a dissipative circuit we solve 
\begin{align}\label{eq:numerics}
    \bm{v}^\infty = \Omega_L \dots \Omega_1 \bm{v}^\infty + \bm{d}_L.
\end{align}
for $\bm{v}^\infty$. This avoids matrix inversion of Eq.~\eqref{eq:SteadyState} and instead solves a linear equation that takes $\mathcal{O}( (4^n )^2)$ rather than $\mathcal{O}( (4^n)^3)$ with $n$ the number of qubits. The main advantage is that the solution to Eq.~\eqref{eq:numerics} gives the exact steady-state of the circuit. Computing the steady-state is still expensive with the matrices that transform the Bloch vectors of size $4^n$ rather than $2^n$ for qubit based simulation. Furthermore, one has to solve a linear equation which scales quadratic in the size of the matrix. 

The conceptual advantages of direct steady-state simulation have led us to build a quantum circuit simulator that allows one to construct a quantum circuit from unitary gates and quantum channels. The simulator then uses the Bloch representation to calculate the exact steady-state of the circuit.

To ensure correctness of the simulator, we first test the standard (non steady-state) quantum circuit simulation. We construct random circuits using random gates and random parameterizations, and compare the simulation output to that of pennylane \cite{bergholm_pennylane_2018}. These simulations verify that both the implementations of the gates and the transpilation into a multi-qubit circuit work correctly \cite{huang_classical_2022}.
Next we compare the direct simulation of the steady-state to layer wise simulation in pennylane. The goal of this comparison is to: \textbf{1.} Verify the results of the simulator that computes the exact steady-state. \textbf{2.} See how many layers are needed when performing a layered simulation. The aim of the circuit is to prepare the Bell state $\ket{\psi} = \frac{1}{\sqrt{2}} \left(\ket{00} + \ket{11} \right) $. The system consists of 3 qubits; two model qubits on which the desired state is prepared and one ancilla. The state $\ket{\psi}$is the ground state of the stabilizer Hamiltonian $ H = -XX-ZZ$ which we choose as a cost function. The circuit consists of Trotterized evolution with the Hamiltonian, followed by coupling to the ancilla with a Heisenberg like interaction. After every jump the ancilla is reset. 
In Fig.~\ref{fig:InfidelityDecayLog} we plot the infidelity between the exact steady-state to the output of the layered simulation. We observe that the infidelity decays exponentially with the number of layers in the circuit. It plateaus at approximately $10^{-3}$. In theory, the infidelity should decay further. A possible explanation are numerical effects.
\begin{figure}
    \centering
    \includegraphics{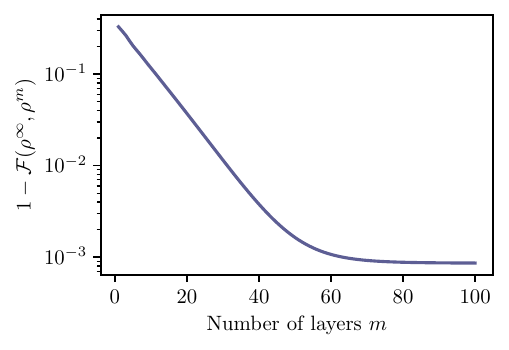}
    \caption{\textbf{Infidelity of layered simulation with steady-state.} We plot the infidelity of the exact steady-state with a state that consists of $m$ layers. The layered simulation exponentially approximates the exact steady-state. At around $60$ layers this decay plateaus with an infidelity of $10^{-3}$. We explain the plateau by numerical effects.}
    \label{fig:InfidelityDecayLog}
\end{figure}

\end{onecolumngrid}
\end{document}